\documentclass[journal,onecolumn, draftclsnofoot, 12pt]{IEEEtran}

\usepackage{cite}

\ifCLASSINFOpdf
\else
\fi

\usepackage{graphicx}
\usepackage{caption}
\usepackage{amsmath,amssymb,amsfonts,algpseudocode}
\usepackage{upgreek}
\usepackage{optidef}
\usepackage{multirow}
\usepackage{array}
\usepackage{stfloats}
\usepackage{amsthm}
\usepackage{lipsum}
\usepackage{mathtools}
\usepackage{cuted}
\usepackage{algorithm}
\usepackage{blindtext}
\usepackage[font=small]{caption}

\newtheorem{lemma}{Lemma}
\def\NoNumber#1{{\def\alglinenumber##1{}\State #1}\addtocounter{ALG@line}{-1}}

\begin{document}
%
\title{Energy-Efficient Communication Networks via Multiple Aerial Reconfigurable Intelligent Surfaces: DRL and Optimization Approach}
%
%
%

\author{Pyae Sone Aung,
        Yu~Min~Park,
        Yan~Kyaw~Tun,~\IEEEmembership{Member,~IEEE,}
        Zhu Han,~\IEEEmembership{Fellow,~IEEE,}
        and~Choong~Seon~Hong,~\IEEEmembership{Senior Member,~IEEE,}
\thanks{Pyae Sone Aung, Yu Min Park, Yan Kyaw Tun, and Choong Seon Hong are with the
Department of Computer Science and Engineering, Kyung Hee University,
Yongin-si, Gyeonggi-do 17104, Rep. of Korea, e-mail:\{pyaesoneaung, yumin0906, ykyawtun7,
cshong\}@khu.ac.kr.}
\thanks{Zhu Han is with the Electrical and Computer Engineering Department,
University of Houston, Houston, TX 77004, and the Department of Computer
Science and Engineering, Kyung Hee University, Yongin-si, Gyeonggi-do
17104, Rep. of Korea, email\{zhan2\}@uh.edu.}}

\maketitle
\begin{abstract}
In the realm of wireless communications in 5G, 6G and beyond, deploying unmanned aerial vehicle (UAV) has been an innovative approach to extend the coverage area due to its easy deployment. Moreover, reconfigurable intelligent surface (RIS) has also emerged as a new paradigm with the goals of enhancing the average sum-rate as well as energy efficiency. By combining these attractive features, an energy-efficient RIS-mounted multiple UAVs (aerial RISs: ARISs) assisted downlink communication system is studied. Due to the obstruction, user equipments (UEs) can have a poor line of sight to communicate with the base station (BS). To solve this, multiple ARISs are implemented to assist the communication between the BS and UEs. Then, the joint optimization problem of deployment of ARIS, ARIS reflective elements on/off states, phase shift, and power control of the multiple ARISs-assisted communication system is formulated. The problem is challenging to solve since it is mixed-integer, non-convex, and NP-hard. To overcome this, it is decomposed into three sub-problems. Afterwards, successive convex approximation (SCA), actor-critic proximal policy optimization (AC-PPO), and whale optimization algorithm (WOA) are employed to solve these sub-problems alternatively. Finally, extensive simulation results have been generated to illustrate the efficacy of our proposed algorithms.
\end{abstract}


%
\IEEEpeerreviewmaketitle

\section{Introduction}


\subsection{Background and Motivations}
As claimed by Cisco Networking Index (CNI), the number of Internet users reached 3.9 billion in 2018 and is anticipated to surpass 5.3 billion by 2023 \cite{cisco2020cisco}. Rapid growth of multimedia devices such as the Internet of Things (IoT), video streaming, online gaming, Virtual Reality (VR) and Augmented Reality (AR) applications, thrives immense challenges for current communication architecture and motivates to discover new ways to enhance spectral efficiency in both academic and industrial fields. Numerous ingenious wireless technologies have been developed in the last several years, which includes deploying unmanned aerial vehicles (UAVs) and Reconfigurable Intelligent Surfaces (RIS) elements.

Recently, UAVs have achieved a great deal of interests to deploy as a communication and computing platforms due to their high mobility and ease of deployment. The emplacement of UAVs can not only save the cost of mobile infrastructure which demands a large budget but also save time for quick on-demand deployment to provide services in rural regions or disaster areas or temporary events such as concerts, stadiums where the infrastructure is difficult to come across. In some scenarios \cite{tun2020energy, hu2019uav}, UAVs are implemented with a multi-access edge computing (MEC) system to deliver the computing resources near to the user equipment (UE) which saves a considerably large amount of time for uploading, computing and downloading tasks.

The newly recent technology called RIS, which is incited from the recent development of meta-surfaces, benefits the wireless communications in extending the coverage range and improving the signal quality at the receiver \cite{elmossallamy2020reconfigurable}. RIS is a man-made meta-surface implemented with low-cost passive elements that can be programmed by integrated electronic circuits to alter the incoming electromagnetic field into the desirable way \cite{huang2019reconfigurable}. Unlike the traditional collaborative communications such as decode-and-forward (DF) and amplify-and-forward (AF), RIS does not require additional power amplifier hence, is more environmental friendly and energy-efficient \cite{liu2021reconfigurable}. Taking into account of its cost efficiency and energy efficiency, RIS technology has acquired a vast attention in 5G, 6G and beyond communications. Furthermore, since RIS structures consist of relatively small hardware components, they can be easily integrated in several communication environments such as along the surfaces of the building \cite{li2019machine}.

\subsection{Challenges and Research Contributions}
When UAVs are considered as communication and computing platforms , there exists several challenges in UAVs' energy consumption as they are energy-constrained devices. On the other side, even though RIS can enhance the spectral efficiency, setting up RIS structures to achieve Line-of-sight (LoS) links between UE and RIS is still quite challenging issues. Taking the advantages of RIS in enhancing spectral efficiency without the requirement of any external power sources with the aid of UAVs to obtain LoS links between UE and RIS, we propose the multiple aerial RISs (ARISs)-assisted system to extend the downlink communication links from the ground base station (BS) to the UEs. In our system, we assume that there is no dominant LoS links between the BS and UEs due to obstacles. The contributions of our paper can be organized as following:
\begin{itemize}
    \item Firstly, we propose the downlink communication system between the BS and UEs, which is assisted by the multiple ARISs to enhance the spectral efficiency for all UEs since the dominant LoS links between BS and UEs are blocked by the obstacles. We assume the BS and ARISs are deployed by the same service operator and thus the BS is responsible for ARIS deployment and controlling the on/off states and the phase shifts for the ARIS reflective elements.
    \item Secondly, we formulate the problem to maximize the energy efficiency of the proposed system by jointly optimizing the ARISs deployment, ARIS reflective elements on/off states, phase shift, and power control. We show that the formulated problem is a mixed integer non-linear programming (MINLP) problem and it is challenging to solve in the polynomial time.
    \item To address this challenge, we decompose our formulated problem into three sub-problems: 1) ARISs deployment problem, 2) joint ARIS reflective elements on/off states and phase shift problem, and 3) power control problem. Then, successive convex approximation (SCA), actor-critic proximal policy optimization (AC-PPO), and whale optimization algorithm (WOA) are proposed to solve these sub-problems, alternatively.
    \item Finally, a comprehensive numerical analysis is integrated to validate efficacy of the overall performance of our proposed algorithms with several benchmark schemes, such as single-ARIS, ARIS with fixed phase shifts (ARIS-NPS), and UAV as relay (UAV-relay) scenarios. We achieve the improvement in average sum-rate by 24$\%$ and 58$\%$ compared to the single-ARIS and the ARIS-NPS scenarios, and 43$\%$ and 72$\%$ increase in energy-efficiency compared to the single-ARIS and the UAV-relay scenarios, respectively. Moreover, our proposed multiple ARISs-assisted system achieves 69$\%$ increase in average sum-rate compared to the multiple RISs-assisted system.
\end{itemize}

The rest of the paper is categorized as follows: we present the related works in Section \ref{relatedworks}. Next, we present our system model and problem formulation in Section \ref{systemmodel}. Afterwards, the solution approach is proposed in Section \ref{solution}, and performance evaluation is performed in Section \ref{evaluation}. Finally, Section \ref{conclusion} concludes our paper.

\section{Related Works}\label{relatedworks}
\subsection{UAV-assisted wireless networks}
An overview on the literature related to UAV-assisted wireless networks are discussed in this section \cite{bucaille2013rapidly, zhan2011wireless, mozaffari2016unmanned, feng2006wlcp2, mozaffari2017wireless,mozaffari2018beyond,wang2019adaptive}. The major strength of UAV in enhancing coverage area, energy-efficiency, and cost-efficiency has received significant attention in recent years \cite{bucaille2013rapidly}. In \cite{zhan2011wireless}, the authors studied to maximize the uplink communication where UAVs are served as relays. In \cite{mozaffari2016unmanned}, the authors studied a single UAV-assisted device-to-device (D2D) communications and analyse how the appropriate UAV's altitude can impact the rate performance and coverage area on the D2D users' density. The authors in \cite{feng2006wlcp2} derived the channel model of the LoS probability for the air-to-ground UAV communications. There exists several works that studied upon UAV deployment \cite{mozaffari2017wireless,mozaffari2018beyond,wang2019adaptive, liu2018energy}. The authors in \cite{mozaffari2017wireless} studied the UAVs deployment for UAV-to-ground communication in arbitrary spatial distribution for network planning to provide wireless services to the ground users and the authors in \cite{mozaffari2018beyond} studied the incorporation between UAVs in 3D cellular network. The work in \cite{wang2019adaptive} studied the adaptive UAV deployment for the dynamic users. The authors in \cite{liu2018energy} studied DRL-based dynamic UAV control instead of static UAV deployment. In all of the above works, UAV is considered either as aerial BS or MEC devices or relays, which results in higher energy consumption.

\subsection{RIS-assisted wireless networks}
An overview on the literature related to RIS-assisted wireless networks are discussed in this section \cite{huang2019reconfigurable, lee2020deep, zhang2020reconfigurable, di2020hybrid, chen2020reconfigurable, ai2021secure, al2022reconfigurable, huang2020reconfigurable, yang2020deep}. In \cite{huang2019reconfigurable}, the authors considered to develop the energy-efficient architecture for the RIS structures in accordance with power allocation and phase shifting values of RIS elements while guaranteeing the individual data rate budget for each user. In \cite{lee2020deep}, the authors proposed the energy-harvesting RIS elements implemented on the facades of the buildings in order to maximize the spectral efficiency while enabling the transmit power control and RIS configuration under the indeterminate wireless channel condition. The authors in \cite{zhang2020reconfigurable} aimed to distinguish the principal relationship between the total sum-rate of multiple users and the required number of RIS reflective elements in wireless communications. They observed the capacity of the system could no longer efficiently rise as the number of RIS elements reaches the upper bound limit. They also investigated how the number of phase shifts can effect the performance on the achievable data rate. The authors in \cite{di2020hybrid} investigated the practical case study between phase shift and finite-sized RIS to maximize the downlink multi-user system. In \cite{chen2020reconfigurable}, the authors studied about RIS elements to eliminate interference between multiple D2D uplink communication network. There has also been several studies on RIS-assisted in the vehicular networks. In \cite{ai2021secure}, the authors investigated the secrecy outage probability upon vehicular-to-vehicular (V2V) and vehicular-to-infrastructure (V2I). The authors in \cite{al2022reconfigurable} aimed to maximize the data rate for each vehicle where the communication links from the road site unit (RSU) is extended by the RIS technology with discrete phase shift. The authors in \cite{huang2020reconfigurable} studied deep reinforcement learning (DRL) based RIS-assisted multi-user downlink multiple input single output (MISO) system. The work in \cite{yang2020deep} considered to improve the secrecy rate of users in RIS-assisted system by constructing DRL-based QoS-aware reward function. All of the aforementioned works only considered the RIS-assisted networks, where RIS elements are either implemented on the ground level or facades of the building, which is still challenging to achieve the dominant LoS communication links between the BS-RIS-users.

\subsection{UAV-RIS-assisted wireless networks}
An overview on the literature related to UAV-RIS-assisted wireless networks are discussed in this section \cite{cao2021reconfigurable, li2021robust, ge2020joint, liu2020machine, shang2021uav, shang2022aerial, li2021aerial, khalili2021resource,  quispe2021joint}. The authors in \cite{cao2021reconfigurable} examined the adaptive RIS-assisted aerial-terrestrial downlink communication system between UAVs and multi-users with respect to RIS elements allocation and reflective coefficients. In \cite{li2021robust}, the authors looked into UAV-user communication with RIS assistance in order to maximize the worst-case secrecy rate by taking into account of the transmitter's power allocation, RIS's beamforming and UAV's trajectory.
The authors in \cite{ge2020joint} proposed the RIS-assisted UAV communications to maximize the received signal power at the ground user by considering the passive and active beamforming and UAV's trajetory. Furthermore, in \cite{liu2020machine}, the authors minimized the energy consumption problem for both orthogonal multiple access (OMA) and non-orthogonal multiple access (NOMA) cases by jointly considering the trajectory for the UAV and passive beamforming of the RIS elements. There also exists several works on ARIS-assisted system \cite{shang2021uav, shang2022aerial}. In \cite{li2021aerial}, the authors considered ARIS-assisted system to satisfy the constraints of ultra-reliable low latency communication (URLLC). The authors in \cite{khalili2021resource} studied the several UAVs-RISs-assisted total transmit power minimization for the heterogeneous networks. They did not consider the energy efficiency of the system. The authors in \cite{quispe2021joint} considered to maximize energy efficiency for a single ARIS-assisted downlink communication for single user. They did not consider for the multiple ARISs-assisted scenario. In this paper, we propose the multiple ARISs in order to maximize the average energy efficiency for the downlink communication between the BS and the UEs. 
\section{System Model}\label{systemmodel}
\begin{figure}[t]
    \centering
	\includegraphics[width=0.7\linewidth]{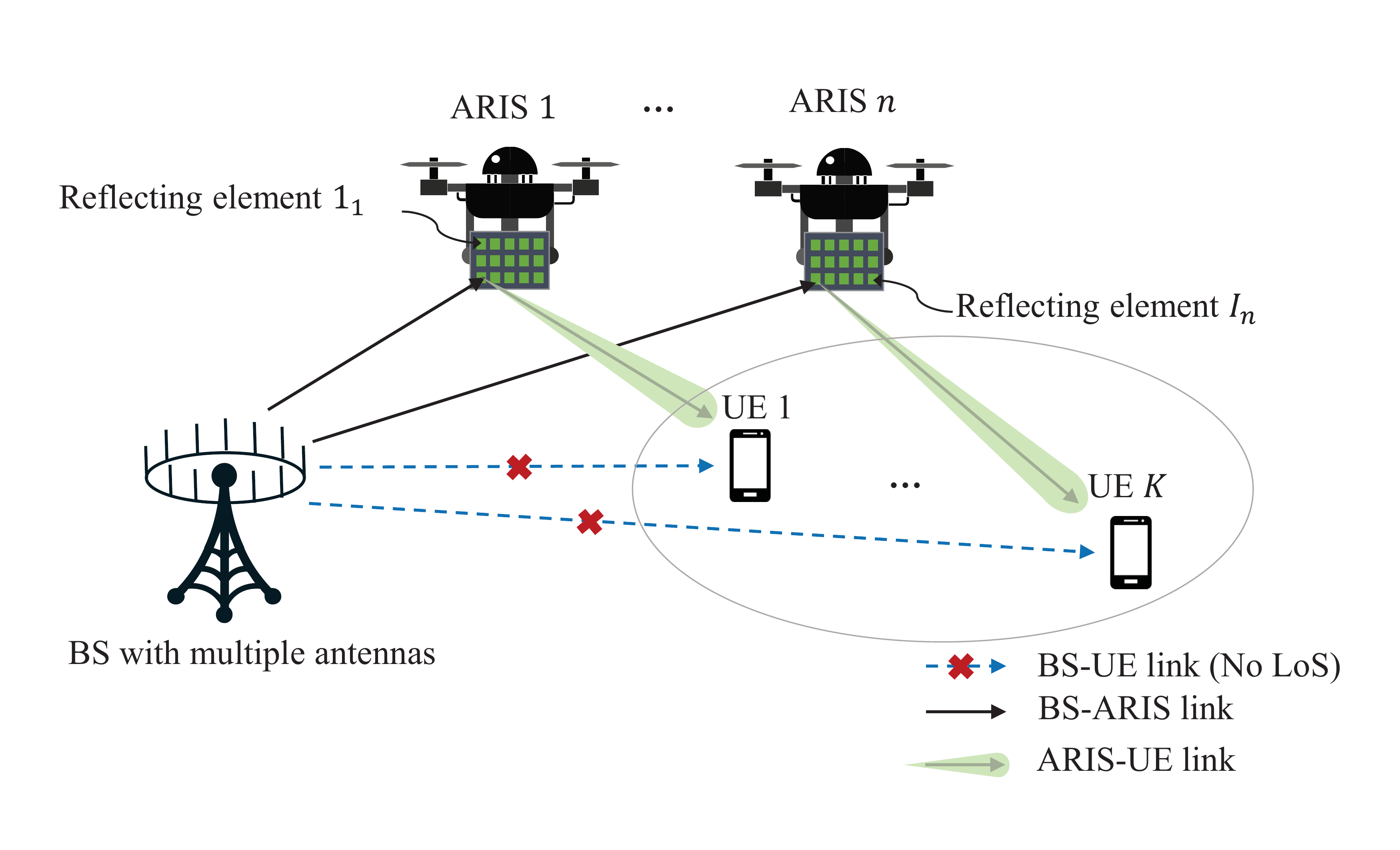}
	\caption{System model for RIS-mounted UAVs.}
	\label{smfig}
\end{figure}
Our system model includes a BS $B$ with multiple antennas, a set $\mathcal{N}$ of $N$ ARISs in which each RIS is implemented on each UAV and each ARIS $n \in \mathcal{N}$ contains an array of $\mathcal{I}_n = [1_n, 2_n, \dots, I_n]$, reflective elements and a set $\mathcal{K}$ of $K$ UEs with single antenna as shown in Fig. \ref{smfig}. The coordinates of the BS is denoted by $\boldsymbol{q}_B = (x_{B},y_{B},z_{B})$, where $z_{B}$ is the height of the BS. Similarly, the positions of each UE $k$ and each ARIS $n$ can be represented as $\boldsymbol{q}_k = (x_{k},y_{k},0)$ and $\boldsymbol{q}_n = (x_{n},y_{n}, z_{n})$ respectively, and $z_{n}$ is the height where the RIS-mounted UAV is hovering. The time horizon of the system can be divided into a discrete set of $\mathcal{T} = [1,2,...,t,...,T]$.

Since ARIS has limited energy, apart from hovering, the ARIS reflective elements need to be turned off when there is no connection in order to reserve excessive energy. Authors in \cite{yang2021energy} and \cite{cao2021reconfigurable} prove that turning off the the whole RIS or some surface area of RIS can preserve energy. In our work, we define $\mathbf{\Delta} \in \mathbb{R}^{|\mathcal{N}|\times|\mathcal{I}_n|}$ as the on/off states matrix for all reflective elements $|\mathcal{I}_n|$ for each ARIS $n$ to decide whether to turn on or off. Therefore, the on/off states of the reflective element $i_n$ in each ARIS $n$ are controlled by the decision variable $\delta_{i_n}$ as follows:
\begin{equation}
\delta_{i_n}[t] =\left \{ \begin{array}{ll}{1,} & {\text { if reflective element $i$ of ARIS $n$ is switched on at time $t$,}}  \\ {0,} & {\text { otherwise.}}\end{array}\right.
\end{equation}

\subsection{Communication Model}
We adopt both direct and indirect communication links between the BS and the UEs. For the direct link, we assume there is no dominant propagation along the LoS signal between the BS and UEs. Therefore, we adopt the Rayleigh fading model and the channel gain for the BS-UE link at time $t$ can be obtained as follows:
\begin{equation}
   \mathbf{H}_{B,k}[t] = \sqrt{\kappa  d_{B,k}^{-\alpha}[t]}\tilde{h},
\end{equation}
where $\kappa$ is the channel gain at the reference distance $1$ m, $\alpha \geq 2$ is the path loss exponent,  $|d_{B,k}[t]|= ||\boldsymbol{q}_B[t]-\boldsymbol{q}_k[t]||$ is the Euclidean distance between the BS and UE $k$ at time $t$, and $\tilde{h}$ is the complex Gaussian random scattering component with zero mean and unit variance.

For the indirect communication, there exists two links: BS-ARIS link and ARIS-UE link, respectively. For the BS-ARIS link, we assume there is only LoS signal between the BS and ARIS, and thus the channel fading here is assumed to experience the Rician channel fading with only LoS components. Therefore, the channel gain between the BS and ARIS $n$ at time $t$ can be defined as:
\begin{equation}
    \mathbf{h}_{B,n}[t] = \sqrt{\kappa d_{B,n}^{-\alpha}[t]} \sqrt{\frac{\hat{R}}{1+\hat{R}}}\mathbf{h}_{B,n}^{\mathrm{LoS}}[t],
\end{equation}
where $\hat{R}$ is the Rician factor, and $|d_{B,n}[t]|= ||\boldsymbol{q}_B[t]-\boldsymbol{q}_n[t]||$ is the distance between the BS to ARIS $n$ at time $t$. $\mathbf{h}_{B,n}^{\mathrm{LoS}}[t]$ is the deterministic LoS component between the BS and ARIS $n$ in correspondence with the azimuth angle-of-arrival (AoA) of the link at time $t$ \cite{cao2021reconfigurable}.
For the ARIS-UE link, there are both LoS and non-line-of-sight (NLoS) propagation between ARISs and UEs. Consequently, the Rician fading model is adopted and the channel gain for the ARIS-UE link at time $t$ can be obtained as follows:
\begin{equation}
\mathbf{h}_{n,k}[t] = \sqrt{\kappa d_{n,k}^{-\alpha}[t]} \sqrt{\frac{\hat{R}}{1+\hat{R}}}\mathbf{h}_{n,k}^{\mathrm{LoS}}[t]+\sqrt{\frac{1}{1+\hat{R}}}\mathbf{h}_{n,k}^{\mathrm{NLOS}},
\end{equation}
where $|d_{n,k}[t]|=||\boldsymbol{q}_n[t] - \boldsymbol{q}_k[t]||$ is the distance between ARIS $n$ and UE $k$ at time $t$. $\mathbf{h}_{n,k}^{\mathrm{LoS}}[t]$ is the deterministic LoS component between ARIS $n$ and UE $k$ corresponding with the azimuth angle-of-departure (AoD) of the link and $\mathbf{h}_{n,k}^{\mathrm{NLOS}}$ is the non-LoS component which follows the identically and independently distributed circularly-symmetric complex Gaussian distribution.

Furthermore, at time $t$, the incident signals are reflected by each reflective element $i$ of ARIS $n$ from the feasible range of phase shift values specified by
\begin{equation}
    \theta_{i_n}[t] = e^{(\frac{2\pi\phi}{2^b})},
\end{equation}
where $\phi$ is the phase shift index, and $b$ is the phase shift resolution in bits \cite{huang2018energy}. Therefore, a vector of $\boldsymbol{\theta}_{i_n}[t] = [\theta_{1_n}[t], \theta_{2_n}[t], \ldots, \theta_{I_n}[t]]$ represents the phase shift values of ARIS $n$. Following that, the reflection coefficient matrix can be denoted by
\begin{equation}
\boldsymbol{\Theta}_n[t] = \operatorname{diag}(\beta_{1_n}e^{j\theta_{1_n}[t]},\beta_{2_n}e^{j\theta_{2_n}[t]},\ldots,\beta_{I_n}e^{j\theta_{I_n}[t]}),
\end{equation}
where $\beta_{i_n} \in [0,1]$ denotes the amplitude reflection coefficient of the $i$-th reflective element of the $n$-th ARIS, and $j$ is the imaginary unit of a complex number. Therefore, the received signal at UE $k$ can be achieved as follows:
\begin{equation}\label{y_k}
    y_k[t] = \left(\mathbf{H}_{B,k}[t] + \sum_{n=1}^{N}\sum_{i=1}^{I_n}\delta_{i_n}[t] \mathbf{h}_{n,k}[t] \boldsymbol{\Theta}_n[t] \mathbf{h}_{B,n}[t]\right) \boldsymbol{x} + \omega_k,
\end{equation}
where $\boldsymbol{x} = \sum_{k=1}^{K} \boldsymbol{g_k}[t] s_k$ is the transmitted signal from the BS with beamforming vector $\boldsymbol{g_k}[t]$ at time $t$, and the unit-power complex based information symbol $s_k$ for UE $k$, while $\omega_k \sim \mathcal{CN}(0,\sigma^2)$ denotes the additive white Gaussian noise (AWGN) at UE $k$. Based on (\ref{y_k}), the signal-to-interference-plus-noise ratio (SINR) received at UE $k$ can be obtained as
\begin{equation}\label{gamma_k}
\gamma_k[t] = \frac{\left|\left(\mathbf{H}_{B,k}[t] +\sum_{n=1}^{N}\sum_{i=1}^{I_n}\delta_{i_n}[t] \mathbf{h}_{n,k}[t] \boldsymbol{\Theta}_n[t] \mathbf{h}_{B,n}[t]\right) \boldsymbol{g_k}[t]\right|^2}{%
\sum_{l=1,l \neq k}^{K}|(\mathbf{H}_{B,k}[t]+\sum_{n=1}^{N}\sum_{i=1}^{I_n}\delta_{i_n}[t] \mathbf{h}_{n,k}[t] \boldsymbol{\Theta}_n[t] \mathbf{h}_{B,n}[t]) \boldsymbol{g_l}[t]|^2 + \sigma^2.}
,
\end{equation}
Afterwards, based on (\ref{gamma_k}), the achievable data rate of UE $k$ can be formulated as follows:
\begin{equation}
    r_k[t] = W \log_{2}(1+\gamma_k[t]),
\end{equation}
where $W$ is the transmission bandwidth available for each UE. Therefore, the sum-rate of all users can be described as follows:
\begin{equation}
    R[t] = \sum_{k=1}^{K} r_k[t].
\end{equation}

\subsection{Power Consumption Model}
In our scenario, we need to take account of the power consumption of ARIS hovering. We assume that ARISs are considered to be hovering at the designated altitude and thus, rotary-wing UAV is adopted. Therefore, the power consumption for the hovering of the rotary-wing UAV, $P_{\mathrm{UAV}}$ can be obtained as follows \cite{jung2021orchestrated}:
\begin{equation}
    P_{\mathrm{UAV}} = \frac{\nu}{8} \varphi \Lambda \eta v_{a}^{3} \varrho+(1+\iota) \frac{\tilde{w}^{3 / 2}}{\sqrt{2 \varphi \eta}},
\end{equation}
which contains two terms: power required to rotate the rotor blades, and power required to endure the induced drag generated by the lift. The symbols $\nu$, $\varphi$, $\Lambda$, $\eta$, $v_a$, and $\varrho$ represent the coefficient of the profile drag, density of the air, rotor solidity, disc area of the rotor, blade angular velocity, and radius of the rotor, respectively. Moreover, $\iota$ and $\tilde{w}$ denote the incremental correction factor, and weight of the aircraft, respectively.

Furthermore, in this work, the BS controls the phase shifts of the ARIS reflective elements. Hence, the total power of the considered multiple ARISs-assisted downlink system includes: 1) transmit power of the BS, 2) circuit power of the each UE $k$, 3) circuit power consumption of ARIS and 4) hovering power of the rotary-wing UAV \cite{yang2021energy}, and is defined as
\begin{equation}
    P[t] = \sum_{k=1}^{K}(\zeta \boldsymbol{g_k}[t]^H \boldsymbol{g_k}[t] + P_k^{\mathrm{cir}}) + \sum_{n=1}^{N} \sum_{i=1}^{I_n} \delta_{i_n}[t] I_n P_{\mathrm{ARIS}} + P_{\mathrm{UAV}},
\end{equation}
where $\zeta = 1 / \mu$ with $\mu$ being the transmit power amplifier efficiency, $P_k^{\mathrm{cir}}$ is the circuit power of each user $k$, and $P_{\mathrm{ARIS}}$ is the power consumption for each ARIS. The transmit signal power of the BS has the constraint as follows:
\begin{equation}
    \mathrm{tr}(\boldsymbol{g}[t]^H\boldsymbol{g}[t]) \leq P_{\max}, \forall{t} \in \mathcal{T},
\end{equation}
where $\mathrm{tr(\mathbf{S})}$ means the trace of square matrix $\mathbf{S}$, $\boldsymbol{g} = [\boldsymbol{g}_1; \dots ; \boldsymbol{g}_k]$ and $P_{\max}$ is the maximum transmission power available at the BS.

\subsection{Problem Formulation}
The main objective of this work is to maximize energy efficiency of the system, i.e., to maximize the average sum-rate $R[t]$ for the UEs under the constraint of the power consumption $P[t]$ of both ARISs and the BS. To accomplish this, we need to jointly optimize the deployment of ARIS, ARIS reflective element on/off states, and phase shift, and power control of the BS. Prior to problem formulation, we define the required constraints as follows:

Each UE $k$ is necessary to fulfill the demand for the specified data rate at time $t$, which is defined as:
\begin{equation}
    r_k[t] \geq r_k^{\min}[t],\forall{k} \in \mathcal{K}, \forall{n} \in \mathcal{N}, \forall{i} \in \mathcal{I}_n, \forall{t} \in \mathcal{T},
\end{equation}
where $r_k^{\min}[t]$ is the minimum data rate requirement for each UE $k$ at time $t$. The accessible phase shift value of $i$-th reflective element $n$-th ARIS at time $t$ should be between 0 to 2$\pi$ as follows:
\begin{equation}
    0 \leq \theta_{i_n}[t] < 2\pi, \forall{n} \in \mathcal{N}, \forall{i} \in \mathcal{I}_n, \forall{t} \in \mathcal{T}.
\end{equation}
A safe distance between two adjacent ARISs is necessary to ensure that the coverage area of each ARIS does not overlap with that of other. Thereby, it can avoid the interference between different ARIS. We denote $d_{\min}[t]$ as the threshold distance between two adjacent ARISs at time $t$, and can be defined as follows:
\begin{equation}
    ||\boldsymbol{q}_{i}[t]-\boldsymbol{q}_{j}[t]||^2 \geq d_{\min}[t], \forall{i, j} \in \mathcal{N}, i \neq j, \forall{t} \in \mathcal{T}.
\end{equation}
Furthermore, each reflective element $i$ of ARIS $n$ can only be either turned on or off at one time slot and can be given as follows:
\begin{equation}
    \delta_{i_n}[t] \in\{0,1\},\forall{n} \in \mathcal{N}, \forall{i} \in \mathcal{I}_n, \forall{t} \in \mathcal{T}.
\end{equation}
\begin{figure}[t!]
    \centering
	\includegraphics[width=0.7\linewidth]{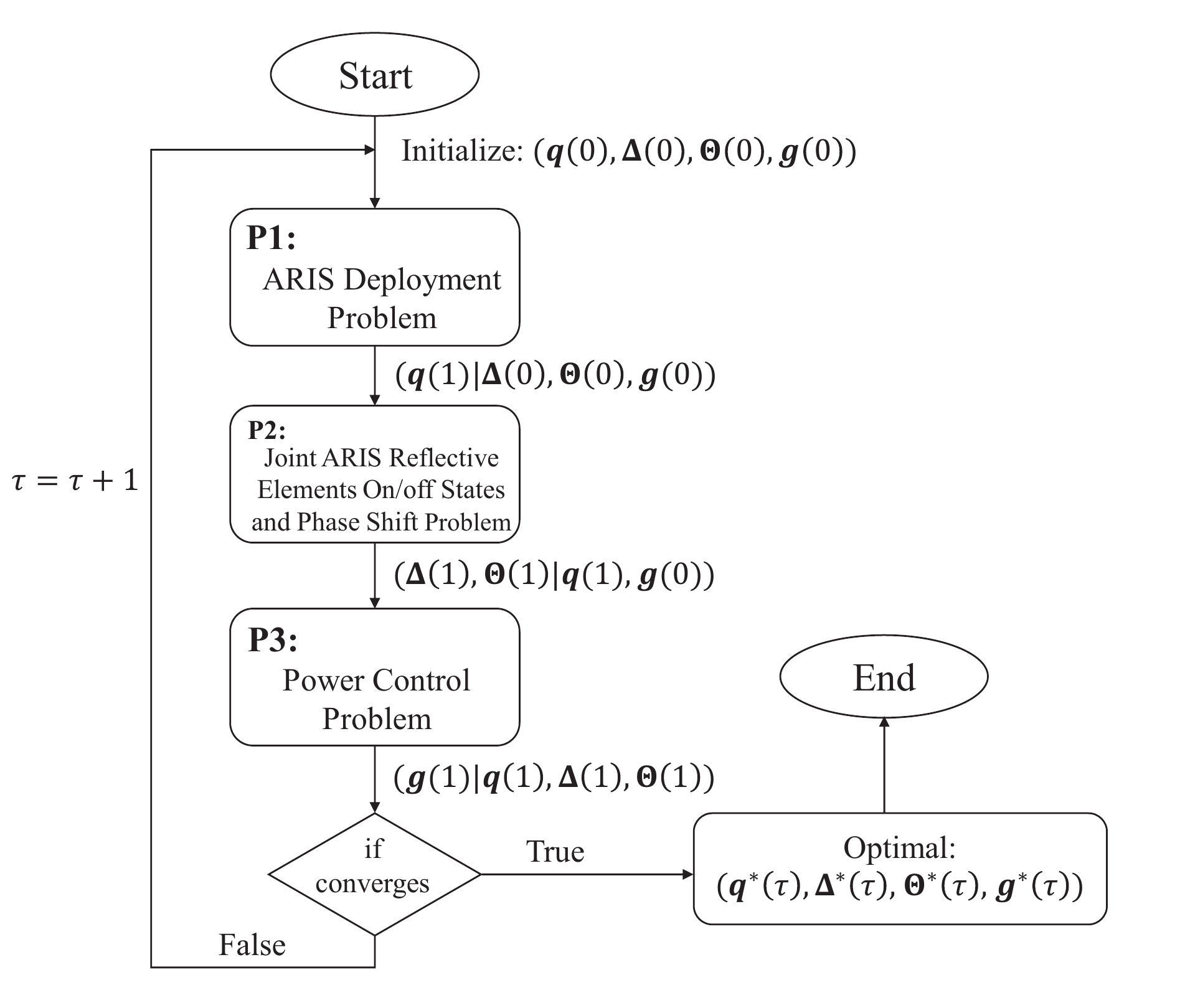}
	\caption{Flow diagram of joint ARIS deployment, ARIS reflective elements on/off states, phase shift and power control problem.}
	\label{flowdiagram}
\end{figure}
Given the above mentioned network characteristics, our optimization problem $\mathcal{E}$ can be mathematically formulated as follows:
\begin{maxi!}|s|[1]
	{\boldsymbol{q}, \boldsymbol{\Delta}, \boldsymbol{\Theta}, \boldsymbol{g}} {\mathcal{E}(\boldsymbol{q}, \boldsymbol{\Delta}, \boldsymbol{\Theta}, \boldsymbol{g})}
	{\label{OF}}{\textbf{P:}}\label{OFa}
	\addConstraint{
	r_k[t] \geq r_k^{\min}[t],\forall{k} \in \mathcal{K}, \forall{n} \in \mathcal{N}, \forall{i} \in \mathcal{I}_n, \forall{t} \in \mathcal{T}}\label{c1}
	\addConstraint{0 \leq \theta_{i_n}[t] < 2\pi, \forall{n} \in \mathcal{N}, \forall{i} \in \mathcal{I}_n, \forall{t} \in \mathcal{T}}\label{c2}
	\addConstraint{
	||\boldsymbol{q}_{i}[t]-\boldsymbol{q}_{j}[t]||^2 \geq d_{\min}[t], \forall{i, j} \in \mathcal{N}, i \neq j, \forall{t} \in \mathcal{T}
	}\label{c3}
	\addConstraint{\mathrm{tr}(\boldsymbol{g}[t]^H\boldsymbol{g}[t]) \leq P_{\max}, \forall{t} \in \mathcal{T}}\label{c4}
	\addConstraint{\delta_{i_n}[t] \in\{0,1\},\forall{n} \in \mathcal{N}, \forall{i} \in \mathcal{I}_n, \forall{t} \in \mathcal{T}.}\label{c6}
\end{maxi!}
\begin{figure*}[t!]
\begin{equation}\label{OFee}
   \mathcal{E}(\boldsymbol{q}, \boldsymbol{\Delta}, \boldsymbol{\Theta}, \boldsymbol{g}) = \frac{1}{T}\sum_{t=0}^{T}\frac{R[t]}{P[t]} = \frac{\sum_{k=1}^{K} W \log_{2}\left(1+\frac{\left|\left(\mathbf{H}_{B,k}[t]+\sum_{n=1}^{N}\sum_{i=1}^{I_n}\delta_{i_n}[t] \mathbf{h}_{n,k}[t] \boldsymbol{\Theta}_n[t] \mathbf{h}_{B,n}[t]\right) \boldsymbol{g_k}[t]\right|^2}{
    \sum_{l=1,l \neq k}^{K}|(\mathbf{H}_{B,k}[t]+\sum_{n=1}^{N}\sum_{i=1}^{I_n}\delta_{i_n}[t] \mathbf{h}_{n,k}[t] \boldsymbol{\Theta}_n[t] \mathbf{h}_{B,n}[t]) \boldsymbol{g_l}[t]|^2 + \sigma^2}\right)}
    {\sum_{k=1}^{K}(\zeta \boldsymbol{g_k}[t]^H \boldsymbol{g_k}[t] + P_k^{\mathrm{cir}}) + \sum_{n=1}^{N} \sum_{i=1}^{I_n} \delta_{i_n}[t] I_n P_{\mathrm{ARIS}} + P_{\mathrm{UAV}}}
\end{equation}
\hrulefill
\end{figure*}
The objective function in (\ref{OFa}) is shown in (\ref{OFee}).
The problem $\textbf{P}$ is a mixed integer non-linear programming (MINLP) problem which is non-convex. Therefore, it is challenging to solve the whole problem in polynomial time. Furthermore, there are couplings in both the objective function and constraints between the ARIS deployment, ARIS reflective elements on/off states and phase shift, and power control of the BS. Therefore, problem $\textbf{P}$ is quite implausible to solve and there is no effective solution approach to deal with these difficulties.

Thus, we first decompose our optimization problem $\textbf{P}$ into three sub-problems, $\textbf{P1:}$ ARIS deployment problem, $\textbf{P2:}$ ARIS reflective elements on/off states and phase shift problem, and $\textbf{P3:}$ power control problem. Then, we solve the sub-problems iteratively until we reach the convergence and the detailed figure of our proposed solution technique is shown in Fig. \ref{flowdiagram}.

\section{Solution Approach}\label{solution}
\subsection{ARIS Deployment Problem}
For the given ARIS reflective elements on/off states $\boldsymbol{\Delta}$, phase shift values $\boldsymbol{\Theta}$ and power control $\boldsymbol{g}$, the sub-problem $\textbf{P1}$ can be represented as follows:
\begin{maxi!}|s|[0]
	{\boldsymbol{q}} {\mathcal{E}(\boldsymbol{q})}
	{\label{OF2}}{\textbf{P1:}}\label{p1}
	\addConstraint{
	r_k[t] \geq r_k^{\min}[t],\forall{k} \in \mathcal{K}, \forall{n} \in \mathcal{N}, \forall{i} \in \mathcal{I}_n, \forall{t} \in \mathcal{T}}\label{c11}
	\addConstraint{
	||\boldsymbol{q}_{i}[t]-\boldsymbol{q}_{j}[t]||^2 \geq d_{\min}[t], \forall{i, j} \in \mathcal{N}, i \neq j, \forall{t} \in \mathcal{T}.}\label{c12}
\end{maxi!}
The objective function of sub-problem $\textbf{P1}$ remains non-concave since $\mathbf{h}_{n,k}[t]$ and $\mathbf{h}_{B,n}[t]$ are complex and non-linear with respect to ARIS deployment $\boldsymbol{q_n}$. To handle this, we use the approximation algorithm for $\mathbf{h}_{n,k}[t]$ and $\mathbf{h}_{B,n}[t]$. Then, we rewrite our sub-problem $\textbf{P1}$ as follows:
\begin{maxi!}|s|[0]
	{\boldsymbol{q}} {\dot{\mathcal{E}}(\boldsymbol{q})}
	{\label{OF211}}{}
	\addConstraint{
	r_k[t] \geq r_k^{\min}[t],\forall{k} \in \mathcal{K}, \forall{n} \in \mathcal{N}, \forall{i} \in \mathcal{I}_n, \forall{t} \in \mathcal{T}}\label{c111}
    \addConstraint{
	||\boldsymbol{q}_{i}[t]-\boldsymbol{q}_{j}[t]||^2 \geq d_{\min}[t], \forall{i, j} \in \mathcal{N}, i \neq j, \forall{t} \in \mathcal{T},}\label{c122}
\end{maxi!}
\begin{figure*}
\begin{equation}\label{dot}
    \dot{\mathcal{E}}(\boldsymbol{q}) =  \sum_{k=1}^{K} W \log_{2}\left(1+\frac{\left(H_{B,k}[t]+\sum_{n=1}^{N}\sum_{i=1}^{I_n}\kappa \delta_{i_n}[t] g_k[t] \mathbf{h}_{ab}^T[t] \mathbf{H'}[t] \mathbf{h}_{ab}[t]\right) }{
    \sum_{l=1,l \neq k}^{K}(H_{B,k}[t]+\sum_{n=1}^{N}\sum_{i=1}^{I_n}\kappa \delta_{i_n}[t] g_l[t] \mathbf{h}_{ab}^T[t] \mathbf{H'}[t] \mathbf{h}_{ab}[t] + \sigma^2)}\right)
\end{equation}
\end{figure*}
where $\dot{\mathcal{E}}(\boldsymbol{q_n})$ is shown in (\ref{dot}), and
\begin{equation*}
    \mathbf{h}_{ab}[t] = \left[ \sqrt{(d_{n,k}[t])^{-\alpha}}, \sqrt{(d_{B,n}[t])^{-\alpha}}  \right]^T,
\end{equation*}
\begin{align*}
\begin{split}
    \mathbf{H'}[t] = &\left[ H_{B,k}^H, (\mathbf{h}_{n,k}^{(\hat{i}-1)}[t])^H  \boldsymbol{\Theta}_n[t] \mathbf{h}_{B,n}^{(\hat{i}-1)}[t] \right]\\
    &\left[ H_{B,k}^H, (\mathbf{h}_{n,k}^{(\hat{i}-1)}[t])^H  \boldsymbol{\Theta}_n[t] \mathbf{h}_{B,n}^{(\hat{i}-1)}[t] \right]^H.
\end{split}
\end{align*}

Next, we introduce the slack variables  $\boldsymbol{a} = \{a[t]\}_{t=1}^T$, $\boldsymbol{b} = \{b[t]\}_{t=1}^T$, and $\boldsymbol{\ddot{r}} = \{\ddot{r}[t]\}_{t=1}^T$, and the problem (\ref{OF211}) is transformed into the following problem as
\begin{maxi!}|s|
	{\boldsymbol{q},\boldsymbol{a},\boldsymbol{b},\boldsymbol{\ddot{r}}} {\ddot{\mathcal{E}}(\boldsymbol{\ddot{r}})}
	{\label{OF22ddot}}{}\label{p111ddot}
	\addConstraint{0 < a[t] \leq \sqrt{(d_{B,n}[t])^{-\alpha}},\forall{t} \in \mathcal{T}}\label{c113}
	\addConstraint{0 < b[t] \leq \sqrt{(d_{n,k}[t])^{-\alpha}},\forall{t} \in \mathcal{T}}\label{c114}
	\addConstraint{\boldsymbol{\tilde{h}}_{ab}^T[t] \boldsymbol{H'}[t] \boldsymbol{\tilde{h}}_{ab}[t] \geq \ddot{r}[t],\forall{t} \in \mathcal{T}}\label{c115}
	\addConstraint{(\ref{c11}),(\ref{c12}).}
\end{maxi!}
where $\ddot{\mathcal{E}}(\boldsymbol{\ddot{r}})$ is given in (\ref{ddot}), and $\boldsymbol{\tilde{h}}_{ab} = [a[t],b[t]]^T$.
\begin{figure*}
\begin{equation}\label{ddot}
    \ddot{\mathcal{E}}(\boldsymbol{\ddot{r}})= \sum_{k=1}^{K} W \log_{2}\left(1+\frac{\left(H_{B,k}[t]+\sum_{n=1}^{N}\sum_{i=1}^{I_n}\kappa \delta_{i_n}[t] g_k[t] \ddot{r}[t]\right) }{
    \sum_{l=1,l \neq k}^{K}(H_{B,k}[t]+\sum_{n=1}^{N}\sum_{i=1}^{I_n}\kappa \delta_{i_n}[t] g_l[t] \ddot{r}[t] + \sigma^2)}\right)
\end{equation}
\hrulefill
\end{figure*}
In order to simplify the derivations, we expand (\ref{c113}) and (\ref{c114}) as follows \cite{li2021robust}:
\begin{equation}\label{abn}
    x_B^2 + x_n[t]^2 + y_B^2 + y_n[t]^2 - 2x_Bx_n[t] - 2y_By_n[t]+ (z_B-z_n[t])^2 - (a[t])^{-\frac{4}{\alpha}} \leq 0,
\end{equation}
\begin{equation}\label{bnk}
    x_n[t]^2 + x_k[t]^2 + y_n[t]^2 + y_k[t]^2 - 2x_n[t]x_k[t] - 2y_n[t]y_k[t]+ (z_n[t]-z_k[t])^2 - (b[t])^{-\frac{4}{\alpha}} \leq 0.
\end{equation}

Still it is discovered that (\ref{abn}) and (\ref{bnk}) are in non-convex feasible regions. Therefore, we apply the SCA method to solve this non-convexity. The SCA approach is advantageous because it allows for the replacement of the original non-convex function with simpler surrogates to achieve a suboptimal solution \cite{scutari2018parallel}. Firstly, to obtain the global upper bound for the concave function, we first utilize the first-order Taylor expansion to find the linear approximation of the function. To do so, firstly, (\ref{p111ddot}) can be transformed into the difference of two concave functions as follows \cite{salman20223to}:
\begin{equation}\label{hml}
   \ddot{\mathcal{E}}(\boldsymbol{\ddot{r}}) \approx \hat{h}(\boldsymbol{\ddot{r}}) - \hat{l}(\boldsymbol{\ddot{r}}),
\end{equation}
where
\begin{equation}\label{hqn}
    \hat{h}(\boldsymbol{\ddot{r}}) = \sum_{k=1}^{K} \log_{2} \left(h_{B,k}[t]+\sum_{n=1}^{N}\sum_{i=1}^{I_n}\kappa \delta_{i_n}[t] g_k[t] \ddot{r}[t] + \sigma^2\right),
\end{equation}
and
\begin{equation}\label{lqn}
    \hat{l}(\boldsymbol{\ddot{r}}) = \sum_{l=1,l \neq k}^{K} \log_{2}\left(h_{B,k}[t]+\sum_{n=1}^{N}\sum_{i=1}^{I_n}\kappa \delta_{i_n}[t] g_l[t] \ddot{r}[t] + \sigma^2 \right).
\end{equation}
Both the functions $\hat{h}(\boldsymbol{\ddot{r}})$ and $\hat{l}(\boldsymbol{\ddot{r}})$ are convex. However the difference between them is neither convex nor concave, as represented in (\ref{hml}). Then, we find the feasible solution $\boldsymbol{\ddot{r}'}$ to problem (\ref{OF22ddot}) by computing the concave lower bound, i.e. the surrogate function of the non-concave objective, specified in (\ref{hml}). By implementing the first-order Taylor expansion to replace the $\hat{l}(\boldsymbol{\ddot{r}})$, we can construct its lower bound as follows:
\begin{equation}\label{lowerbound}
    \ddot{\mathcal{E}}(\boldsymbol{\ddot{r}, \ddot{r}'}) = \hat{h}(\boldsymbol{\ddot{r}}) - \hat{\hat{l}}\left((\boldsymbol{\ddot{r}, \ddot{r}'})\right),
\end{equation}
where
\begin{equation}\label{surrogate}
    \hat{\hat{l}}\left((\boldsymbol{\ddot{r}, \ddot{r}'})\right) \triangleq \hat{l}(\boldsymbol{\ddot{r}'}) - \nabla \hat{l}(\boldsymbol{\ddot{r}'})(\boldsymbol{\ddot{r}}-\boldsymbol{\ddot{r}'}),
\end{equation}
where $\nabla \hat{l}(\boldsymbol{\ddot{r}'})$ is the gradient of the $\hat{l}(\boldsymbol{\ddot{r}})$ at the given point $\boldsymbol{\ddot{r}'}$, and $\hat{\hat{l}}\left((\boldsymbol{\ddot{r}, \ddot{r}'})\right)$ represents the first-order Taylor's approximation of $\hat{l}(\boldsymbol{\ddot{r}})$ near $\boldsymbol{\ddot{r}'}$ in the feasible area of the solution space. The gradient for ARIS $n$ can be expressed as follows:
\begin{equation}
     \nabla_n \hat{l}(\boldsymbol{\ddot{r}'}) = \frac{\partial \hat{l}(\boldsymbol{\ddot{r}'})}{\partial \ddot{r}'} = \frac{1}{\ln2}\sum_{l=1,l \neq k}^{K} \frac{\sum_{n=1}^{N}\sum_{i=1}^{I_n}\kappa \delta_{i_n}[t] g_l[t]}{H_{B,k}[t]+\sum_{n=1}^{N}\sum_{i=1}^{I_n}\kappa \delta_{i_n}[t] g_l[t] \ddot{r}[t] + \sigma^2}.
\end{equation}
The surrogate function given in (\ref{surrogate}) is concave. Next, the upper bound of function $\hat{l}(\boldsymbol{\ddot{r}})$ may also be found using the first-order Taylor's expansion.
\begin{lemma}
The first-order Taylor approximation provides the global upper bound of a concave function or the global lowest bound of a convex function.
\end{lemma}
\begin{proof}
Initially, we define the first-order Taylor series as follows:
\begin{equation}
    f(x_0) + f'(x_0)(x-x_0).
\end{equation}
Afterwards, we have
\begin{equation}\label{upperbound}
    \hat{l}(\boldsymbol{\ddot{r}}) \leq \hat{l}(\boldsymbol{\ddot{r}'}) + \nabla \hat{l}(\boldsymbol{\ddot{r}'})(\boldsymbol{\ddot{r}}-\boldsymbol{\ddot{r}'}).
\end{equation}
Therefore, we can derive the observations by examining (\ref{hml}), (\ref{lowerbound}), and (\ref{upperbound}) as follows:
\begin{equation}\label{observation}
\begin{aligned}
\ddot{\mathcal{E}}(\boldsymbol{\ddot{r}}) &= \hat{h}(\boldsymbol{\ddot{r}}) - \hat{l}(\boldsymbol{\ddot{r}}) \\
& \geq \hat{h}(\boldsymbol{\ddot{r}})-\left\{\hat{l}(\boldsymbol{\ddot{r}'}) + \nabla \hat{l}(\boldsymbol{\ddot{r}'})(\boldsymbol{\ddot{r}}-\boldsymbol{\ddot{r}'})\right\} \\
& \geq \hat{h}(\boldsymbol{\ddot{r}})-\hat{l}(\boldsymbol{\ddot{r}'}) - \nabla \hat{l}(\boldsymbol{\ddot{r}'})(\boldsymbol{\ddot{r}}-\boldsymbol{\ddot{r}'}) \\
&= \ddot{\mathcal{E}}(\boldsymbol{\ddot{r}, \ddot{r}'}),
\end{aligned}
\end{equation}
where (\ref{observation}) denotes that the surrogate function provides the lower bound of the original function. As a result, at point $\boldsymbol{\ddot{r}'}$, i.e., $\ddot{\mathcal{E}}(\boldsymbol{\ddot{r}, \ddot{r}'})|_{\ddot{r}=\ddot{r}'}=\ddot{\mathcal{E}}(\boldsymbol{\ddot{r}'})$, the two functions are tangent to each other. Thereby, our objective function of sub-problem (\ref{OF22ddot}) has the lower bound function as obtained in (\ref{observation}).
\end{proof}
Consequently, we replace our objective function in problem (\ref{OF22ddot}) which is non-convex, by its surrogates as presented in (\ref{lowerbound}). Furthermore, we take the first-order Taylor expansions of $(a[t])^{-\frac{4}{\alpha}}, (b[t])^{-\frac{4}{\alpha}}$, and $\mathbf{\tilde{h}}_{ab}^T[t] \mathbf{H'}[t] \mathbf{\tilde{h}}_{ab}[t]$ at the given feasible points $\boldsymbol{a_0} = \{a_0[t]\}_{t=1}^T$, $\boldsymbol{b_0} = \{b_0[t]\}_{t=1}^T$, and $\mathbf{\tilde{H_0}}_{ab} = \{\mathbf{\tilde{h_0}}_{ab}[t]\}_{t=1}^T$ are expressed as follows:
\begin{equation}\label{ta}
    (a[t])^{-\frac{4}{\alpha}} \geq (a_0[t])^{-\frac{4}{\alpha}}- \frac{4}{\alpha}(a_0[t])^{-\frac{4}{\alpha}-1}(a[t]-a_0[t]),
\end{equation}
\begin{equation}\label{tb}
    (b[t])^{-\frac{4}{\alpha}} \geq (b_0[t])^{-\frac{4}{\alpha}}- \frac{4}{\alpha}(b_0[t])^{-\frac{4}{\alpha}-1}(b[t]-b_0[t]),
\end{equation}
\begin{equation}\label{th}
    \mathbf{\tilde{h}}_{ab}^T[t] \mathbf{H'}[t] \mathbf{\tilde{h}}_{ab}[t] \geq-\mathbf{\tilde{h_0}}_{ab}^T[t] \mathbf{H'}[t] \mathbf{\tilde{h_0}}_{ab}[t]+2\Re\left[\mathbf{\tilde{h_0}}_{ab}^T[t] \mathbf{H'}[t]\mathbf{\tilde{h}}_{ab}[t]\right].
\end{equation}
By combining (\ref{abn}) and (\ref{ta}), (\ref{bnk}) and (\ref{tb}), we get
\begin{align}
&\begin{aligned}\label{finalabn}
    &x_B^2 + x_n[t]^2 + y_B^2 + y_n[t]^2 - 2x_Bx_n[t] - 2y_By_n[t] + \\
    & (z_B-z_n[t])^2 - (1+\frac{4}{\alpha})(a_0[t])^{-\frac{4}{\alpha}}  + \frac{4}{\alpha}(a_0[t])^{-\frac{4}{\alpha}-1} a[t] \leq 0,
\end{aligned}\\
&\begin{aligned}\label{finalbnk}
    &x_n[t]^2 + x_k[t]^2 + y_n[t]^2 + y_k[t]^2 - 2x_n[t]x_k[t] - 2y_n[t]y_k[t]+\\
    &(z_n[t]-z_k[t])^2 - (1+\frac{4}{\alpha})(b_0[t])^{-\frac{4}{\alpha}}  + \frac{4}{\alpha}(b_0[t])^{-\frac{4}{\alpha}-1} b[t] \leq 0.\\
\end{aligned}
\end{align}

Similarly, we apply the first-order Taylor expansion to convert $||\boldsymbol{q}_{i}[t]-\boldsymbol{q}_{j}[t]||^2$ in constraint (\ref{c12}) to a linear function since it is a convex function with respect to $\boldsymbol{q}_{i}$ and $\boldsymbol{q}_{j}$. This can be expressed as follows:
\begin{equation}
    ||\boldsymbol{q}_{i}[t]-\boldsymbol{q}_{j}[t]||^2 \geq2(\boldsymbol{q}_{i}[t-1]-\boldsymbol{q}_{j}[t-1])^T(\boldsymbol{q}_{i}[t]-\boldsymbol{q}_{j}[t])-||\boldsymbol{q}_{i}[t-1]-\boldsymbol{q}_{j}[t-1]||^2.
\end{equation}
Afterwards, we can denote the above equation as follows:
\begin{equation}\label{transform}
     G_0[t-1](\boldsymbol{q}_{i}[t]-\boldsymbol{q}_{j}[t]) \triangleq2(\boldsymbol{q}_{i}[t-1]-\boldsymbol{q}_{j}[t-1])^T(\boldsymbol{q}_{i}[t]-\boldsymbol{q}_{j}[t])-||\boldsymbol{q}_{i}[t-1]-\boldsymbol{q}_{j}[t-1]||^2.
\end{equation}
Finally, we can substitute (\ref{transform}) into (\ref{c12}), and problem (\ref{OF22ddot}) can be rewritten as follows:
\begin{mini!}|s|[0]
	{\boldsymbol{q}, \boldsymbol{a},\boldsymbol{b},\boldsymbol{\ddot{r}}} {-\ddot{\mathcal{E}}(\boldsymbol{\ddot{r}, \ddot{r}'})}
	{\label{OFdde}}{}
	\addConstraint{
	\ddot{r}[t]+\boldsymbol{\tilde{h}_0}_{ab}^T[t] \boldsymbol{H'}[t] \boldsymbol{\tilde{h}_0}_{ab}[t]-2\Re\left[\boldsymbol{\tilde{h}_0}_{ab}^T[t] \boldsymbol{H'}[t]\boldsymbol{\tilde{h}}_{ab}[t]\right] \leq 0, \forall{t} \in \mathcal{T}
	}\label{cdde1}
	\addConstraint{
	G_0[t-1](\boldsymbol{q}_{i}[t]-\boldsymbol{q}_{j}[t]) \geq d_{\min}[t], \forall{i, j} \in \mathcal{N}, i \neq j, \forall{t} \in \mathcal{T}
	}\label{cdde2}
	\addConstraint{(\ref{c11}), (\ref{finalabn}), (\ref{finalbnk}).}\label{c1111}
\end{mini!}
Problem (\ref{OFdde}) becomes a convex optimization problem, which we can solve by using CVXPY solver in python programming. The overall algorithm of the SCA method is shown in Algorithm \ref{algo1}.
\begin{algorithm}[t] \caption{SCA algorithm for ARIS deployment}\label{algo1}
	\begin{algorithmic}[1]
		\renewcommand{\algorithmicrequire}{\textbf{Input:}}
		\Require Initial feasible points $\{\boldsymbol{q}^0, \boldsymbol{a}^0, \boldsymbol{b}^0\}$, $r_k^{\mathsf{min}}[t]$, $d_{\min}[t]$, iteration index $\hat{i} = 0$, $\hat{i}_{max}$, stopping criterion $\varepsilon_1$.
		\Repeat
		    \State Set $\hat{i} \leftarrow \hat{i} + 1$.
		    \State Update $\boldsymbol{q}^{\hat{i}}, \boldsymbol{a}^{\hat{i}}, \boldsymbol{b}^{\hat{i}}$ with given $\boldsymbol{q}^{\hat{i}-1}, \boldsymbol{a}^{\hat{i}-1}, \boldsymbol{b}^{\hat{i}-1}$.
		    \State Acquire $\ddot{\mathcal{E}}(\boldsymbol{\ddot{r}, \ddot{r}'}) = \hat{h}(\boldsymbol{\ddot{r}}) - \hat{\hat{l}}\left((\boldsymbol{\ddot{r}, \ddot{r}'})\right)$ based on (\ref{lowerbound}).
		    \State Solve (\ref{OFdde}) to obtain $\boldsymbol{\ddot{r}}^{\hat{i}}$.
		\Until $|\ddot{\mathcal{E}}(\boldsymbol{\ddot{r}}^{\hat{i}})-\ddot{\mathcal{E}}(\boldsymbol{\ddot{r}}^{\hat{i}-1})| \leq \varepsilon_1$ or $\hat{i} > \hat{i}_{max}$.
        \renewcommand{\algorithmicrequire}{\textbf{Output:}}
    	\Require Optimal ARIS deployment $\boldsymbol{q}^*$.
	\end{algorithmic}
\end{algorithm}
\subsection{Joint ARIS Reflective Elements On/off States and Phase Shift Problem}
For the given ARIS deployment $\boldsymbol{q}$ and power control $\boldsymbol{g}$, the sub-problem $\textbf{P2}$ can be represented as follows:
\begin{maxi!}|s|[0]
	{\boldsymbol{\Delta, \Theta}} {\mathcal{E}(\boldsymbol{\Delta,\Theta})}
	{\label{OF1}}{\textbf{P2:}}
	\addConstraint{
	r_k[t] \geq r_k^{\min}[t], \forall{k} \in \mathcal{K}, \forall{n} \in \mathcal{N}, \forall{i} \in \mathcal{I}_n, \forall{t} \in \mathcal{T}}\label{c21}
	\addConstraint{0 \leq \theta_{i_n}[t] < 2\pi, \forall{n} \in \mathcal{N}, \forall{i} \in \mathcal{I}_n, \forall{t} \in \mathcal{T}}\label{c22}
	\addConstraint{\delta_{i_n}[t] \in\{0,1\},\forall{n} \in \mathcal{N}, \forall{i} \in \mathcal{I}_n, \forall{t} \in \mathcal{T}.}\label{c24}
\end{maxi!}
This problem is still mixed-integer, non-convex, and quite challenging to solve in polynomial time, since the information of the environment is unknown. Moreover, the real-time ARIS reflective elements on/off states requires extensive computation and hardware cost, and conventional optimization methods cannot be applied. The exhaustive search method can be used to find the optimal solution, however it is impractical for large-scale networks. Due to these reasons, we propose DRL approach to solve sub-problem $\textbf{P2}$. The reason we do not apply DRL for the whole optimization problem is that the action spaces combined for all ARIS deployment, ARIS reflective elements on/off states, phase shift, and power control matrices will be too large and demands high computational cost. Here, we implement Actor-Critic Proximal Policy Optimization (AC-PPO) \cite{schulman2017proximal} as it always provides an improved policy by using data that are currently accessible by the agent and thereby ensuring data efficiency and reliable performance. It could also be utilised in the environments where action spaces are discrete or continuous. Typically, since DRL is interpreted as Markov Decision Process (MDP), we first need to define state space $\mathcal{S}$, action space $\mathcal{A}$ and reward $\tilde{R}$.
\subsubsection{State Space}
For each state at time $t$, $s_t \in \mathcal{S}$ can be expressed as the tuples of the users' locations and ARISs' locations, the channel gain of the direct link, the channel gain of the ARIS-UE link and BS-ARIS link, and power control at time $t$, respectively, and can be represented by $\boldsymbol{s}_t = \{\boldsymbol{q}_k[t], \boldsymbol{q}_n[t],\mathbf{H}_{B,k}[t],  \mathbf{h}_{n,k}[t], \mathbf{h}_{B,n}[t], \boldsymbol{g}_k[t], \forall{k} \in \mathcal{K}, \forall{n} \in \mathcal{N}, \forall{i} \in \mathcal{I}_n \}$.

\subsubsection{Action space}
The action at time $t$, $a_t \in \mathcal{A}$ contains the combination of the ARIS reflective elements on/off states variable $\delta_{i_n}[t]$, and phase shift values $\theta_{i_n}[t]$ at time $t$, and can be denoted as $\boldsymbol{a}_t = \{\delta_{i_n}[t], \theta_{i_n}[t], \forall{n} \in \mathcal{N}, \forall{i} \in \mathcal{I}_n\}$

\subsubsection{Reward}
Since the goal of our system is to maximize the energy efficiency, our reward function is defined as
\begin{equation}
\tilde{R}_t(s_t|a_t) =\left \{ \begin{array}{ll}{-1,} & {\text { if $\sum_{k=1}^{K} r_k[t] < r_k^{\min}[t]$,}}  \\ \text {$\mathcal{E}(\boldsymbol{\Delta,\Theta})$,} & {\text { otherwise.}} \end{array}\right.
\end{equation}
As shown in Fig \ref{ACPPO}, in our AC-PPO algorithm, the states information, $s_t$ from the environment is obtained by the agent at the BS, and the agent observes and monitors the status of the location of the users and ARISs, channel gain of the links, and power control for each user. The agent includes the actor model and the critic model\cite{lim2020federated}. The actor model has the stochastic policy model $\pi_\psi(a_t|s_t)$ with its own parameter $\psi$ and learns to take which action under the observation of the input states. The policy $\pi_\psi(a_t|s_t)$ takes the observed states $s_t$ from the environment as an input and suggests actions $a_t$ to take as an output, and calculates the immediate reward $\tilde{R}_t(s_t|a_t)$ depending on the action taken. The reward then provides as feedback to the agent, and the new state information $s(t+1)$ is obtained. Taking into account of the requirements for the users, under given policy $\pi_\psi(a_t|s_t)$ and reward function $\tilde{R}_t(s_t|a_t)$, the cumulative discounted reward function at time $t$ can be denoted as follows:
\begin{equation}
    V^{\pi_\psi}(s_t) = \mathbb{\hat{E}}_t\left[\sum_{t'=t}^{T-1}\xi^{t'-t} \tilde{R}_{t'}(s_{t'}|a_{t'})\right], \forall{s_t} \in \mathcal{S},
\end{equation}
where $0 < \xi < 1$ is the discount factor to prevent the total reward from reaching to infinity.

Moreover, the critic model contains the advantage function, $\hat{A_t}$ which is the estimate of the relative value of the selected action in the current state is defined as \cite{mnih2016asynchronous}:
\begin{equation}{\label{A}}
    \hat{A_t} = V^{\pi_\psi}(s_t) - b(s_t), \forall{s} \in \mathcal{S},
\end{equation}
where $b(s_t)$ is the baseline estimate value function which provides the estimate of the discounted return starting from the current state $s_t$.

The surrogate objective function of AC-PPO is to find the policy that maximizes the total rewards from the environment and can be expressed as follows \cite{schulman2017proximal}: 
\begin{equation}\label{L}
    L^{CLIP}(\psi) = \mathbb{\hat{E}}_t\left[\min \left(r_t(\psi) \hat{A}_t, \operatorname{clip}(r_t(\psi), 1-\epsilon, 1+\epsilon) \hat{A}_t\right)\right],
\end{equation}
where $${r}_{t}(\psi) = \frac{\pi_\psi(a_t|s_t)}{\pi_{\psi_{old}}(a_t|s_t)},$$ means the probability ratio. Given the states and actions, $r_t(\psi) > 1$ is the action is more plausible currently than it was in the old version of the policy, and $0 < r_t(\psi) < 1$ if it is less plausible, and $\epsilon$ is the clipping parameter. The clipping part of the objective function ensures that the PPO does not always favor actions with positive advantage and/or consistently avoid actions with negative advantage. The overall algorithm of the AC-PPO algorithm is described in Algorithm \ref{algo2}.

\begin{algorithm}[t] \caption{AC-PPO algorithm for ARIS reflective elements on/off states and phase shift}\label{algo2}
	\begin{algorithmic}[1]
		\renewcommand{\algorithmicrequire}{\textbf{Input:}}
		\Require Network states $s_t$, learning rate, discount factor $\xi$, clipping parameter $\epsilon$;
		\State \textbf{Initialization} Base policy $\pi_\psi(a_t|s_t)$ with random parameters $\psi$ and clipping parameter $\epsilon$ and initial value function $V^{\pi_\psi}(s_t)$
		\For{$k \in K$}
    		\For{each episode $\hat{t} \in \hat{T}$}
    		    \State Collect the network observations: ARIS deploy-
    		    \NoNumber ments $\boldsymbol{q}$ from Algorithm \ref{algo1} and power control $\boldsymbol{g}$
    		    \NoNumber from Algorithm \ref{algo3} to achieve the initial state $s_0$
    		    \For{each $t \in T$}
    		        \State Forward the network states $s_t \in \mathcal{S}$ to the AC-
    		        \NoNumber PPO algorithm
        		    \State Observe the input states $s_t$ and run the actor
        		    \NoNumber network
        		    \State Select action $a_t \in \mathcal{A}$ based on policy $\pi_\psi(a_t|s_t)$
        		    \State Obtain the reward $\tilde{R}_t(s_t|a_t)$ and $s_{t+1}$
        		    \State Calculate the probability ratio, ${r}_t$
        		    \State Compute $\hat{A_t}$ based on current $V^{\pi_\psi}(s_t)$ at the \NoNumber critic network according to (\ref{A})
        		    \State Compute $L^{CLIP}(\psi)$ according to (\ref{L})
        		    \State Update $\pi_{\psi_{old}}\leftarrow \pi_\psi$
        		 \EndFor
        	\EndFor
    	\EndFor
        \renewcommand{\algorithmicrequire}{\textbf{Output:}}
    	\Require Optimal AC-PPO network with $\boldsymbol{{\Delta}^*}, \boldsymbol{{\Theta}^*}$.
	\end{algorithmic}
\end{algorithm}

\begin{figure*}
    \centering
	\includegraphics[width=\textwidth]{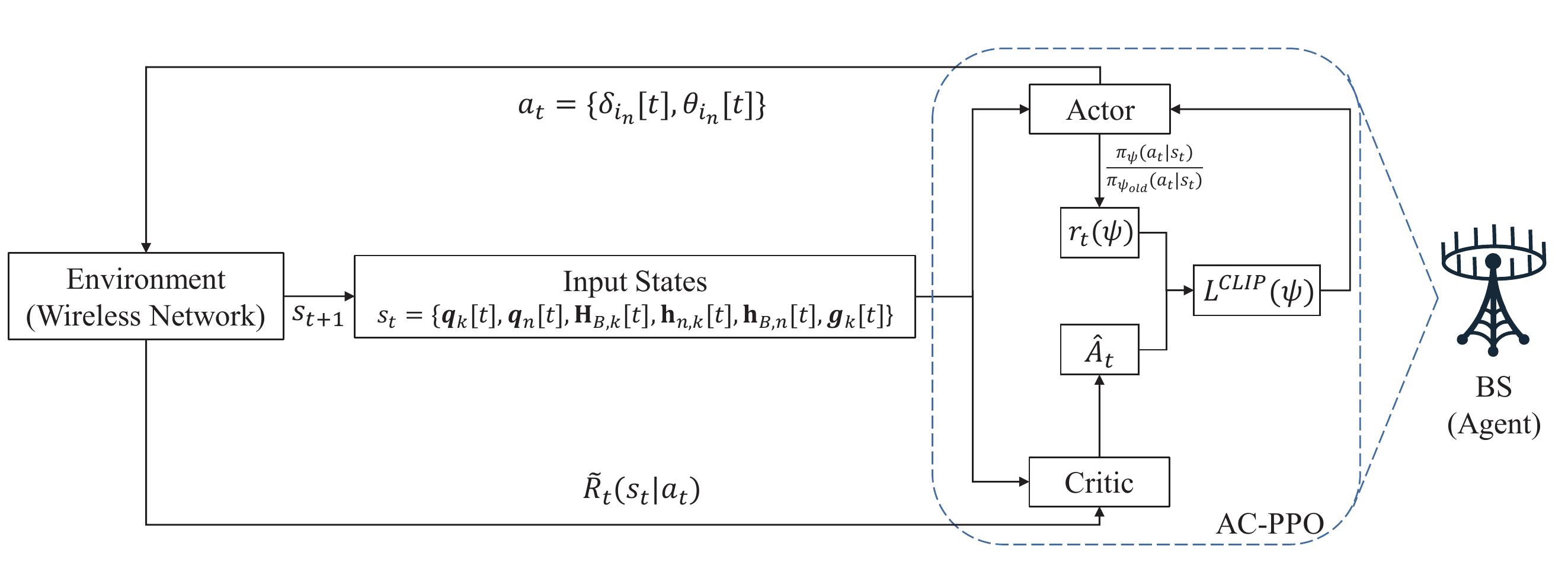}
	\caption{AC-PPO Algorithm for joint ARIS reflective elements on/off states and phase shift.}
	\captionsetup{justification=centering}
	\label{ACPPO}
\end{figure*}

\subsection{Power Control Problem}
For the fixed ARIS deployment $\boldsymbol{q}$, ARIS reflective elements on/off states $\boldsymbol{\Delta}$, and phase shift $\boldsymbol{\Theta}$, sub-problem $\textbf{P3}$ can be represented as follows:
\begin{maxi!}|s|[0]
	{\boldsymbol{g}} {\mathcal{E}(\boldsymbol{g})}
	{\label{OF3}}{\textbf{P3:}}
	\addConstraint{
	r_k[t] \geq r_k^{\min}[t], \forall{k} \in \mathcal{K}, \forall{n} \in \mathcal{N}, \forall{i} \in \mathcal{I}_n, \forall{t} \in \mathcal{T}}\label{c31}
	\addConstraint{\mathrm{tr}(\boldsymbol{g}[t]^H\boldsymbol{g}[t]) \leq P_{\max}, \forall{t} \in \mathcal{T}}.\label{c32}
\end{maxi!}
Sub-problem $\textbf{P3}$ is still a non-convex and NP-hard problem due to constraint (\ref{c31}). Therefore, it is challenging to obtain the solutions in the polynomial time. Therefore, we adopt Whale Optimization Algorithm (WOA) to solve sub-problem $\textbf{P3}$. The WOA is a meta-heuristic algorithm which mimics the whales hunting strategy. The WOA has substantial advantages. First, unlike gradient-based algorithms, which involve computing and updating the gradients and step size throughout every iteration of the optimization process, WOA allows for such computation to be relaxed. Second, WOA is not influenced by the initial feasible solutions, which might have a significant impact on the convergence. Therefore, it has recently gained popularity among research community due to it being efficient optimizer. The WOA algorithm includes two states: 1) the exploitation state (the encircling prey method and spiral bubble-net attacking method), and 2) the exploration state (the searching prey method). The detail explanation of each state can be further described in the following subsections \cite{mirjalili2016whale, mafarja2018whale, pham2020whale}.

\subsubsection{\textbf{Exploitation State}}
The exploitation state of WOA includes two fundamental methods: the encircling prey method, and the spiral bubble-net attack method, which are discussed as follows:

\textbf{Encircling Prey Method}. Once the whales detect the location of their preys, they encircle them. Theoretically, the location of the prey is unknown in the search space, therefore, WOA assumes that the current best search agent is the target prey (optimum or close to optimum). The other whales (search agents) update their locations towards to the best search agent. This behaviour can be mathematically implemented as follows \cite{mirjalili2016whale}:
\begin{equation}\label{ep1}
    \vec{D}=\left|\vec{C} \cdot \vec{g}^{*}(\hat{j})-\vec{g}(\hat{j})\right|,
\end{equation}
\begin{equation}\label{ep2}
    \vec{g}(\hat{j}+1)=\vec{g}^{*}(\hat{j})-\vec{A} \cdot \vec{D},
\end{equation}
 where $\vec{g}^{*}$ is the location of the best search agent, $\hat{j}$ is the current iteration, $\left|\cdot\right|$ is the absolute value. $\vec{C}$ and $\vec{A}$ are coefficient vectors, and are computed as follows:
\begin{equation}\label{vecA}
    \vec{A}=2 \vec{a} \cdot \vec{r}-\vec{a},
\end{equation}
\begin{equation}\label{vecC}
    \vec{C}=2 \cdot \vec{r},
\end{equation}
 where $\vec{r}$ is the random vector between $0$ to $1$, and $\vec{a}$ is the control parameter vector linearly declining from $2$ to $0$ over the iterations, both in exploitation and exploration states. The aim of (\ref{vecA}) and (\ref{vecC}) is to balance between exploitation and exploration. When $A \geq 1$, WOA will perform exploration, and exploitation is done when $A < 1$.
 
 \textbf{Spiral Bubble-net Attack Method}. 
This method combines both shrinking encircling mechanism and spiral movement mechanism of whales. Its purpose is to update the new location to fall between the current agent's location and the best search agent. To mimic the helical shape movement of the whales, the equation can be expressed as
\begin{equation}\label{bb1}
    \vec{D^{\prime}}=\left|\vec{g}^{*}(\hat{j})-\vec{g}(\hat{j})\right|,
\end{equation}
\begin{equation}\label{bb2}
    \vec{g}(\hat{j}+1) = \vec{D^{\prime}} \cdot e^{b j} \cdot \cos (2 \pi l)+\vec{g}^{*}(\hat{j}),
\end{equation}
where $\vec{D^{\prime}}$ indicates the distance between the current search agent and the target prey. Moreover, $b$ is the constant for defining the shape of the logarithmic spiral, and $l$ is the random number between $-1$ and $1$. Here, coefficient vector $\vec{A}$ is updated by setting random values in $[-1,1]$.

\begin{algorithm}[t] \caption{WOA for power control}\label{algo3}
	\begin{algorithmic}[1]
	    \renewcommand{\algorithmicrequire}{\textbf{Input:}}
		\Require Current power control $\boldsymbol{g}$, given $\boldsymbol{q}$, $\boldsymbol{\Delta}$, and $\boldsymbol{\Theta}$;
		\State \textbf{Initialization} At iteration $\hat{j}=1$, initialize the total number of whale population $g_u,$ where $u = \{1,\dots,U\}$, and maximum number of iteration $\hat{j}_{\max}$.
		\State According to (\ref{fitness}), calculate the fitness of the search agents $g_u$ and identify the best search agent $\vec{g}^{*}(0)$.
	    \Repeat
    		\For{$u \leftarrow 1$ to $U$ (the number of whales)}
    		    \State Update $a, A, C, l$ and $p$.
    		    \If{$p < 0.5$}
    		        \If{$|A| < 1$}
    		            \State Update $\vec{D}$ by (\ref{ep1}) and $\vec{g}$ by (\ref{ep2}).
    		        \Else
        		        \State Select a random $\vec{g}_{\mathrm{rand}}$ and update $\vec{D}$ by \NoNumber(\ref{sp1}).
        		        \State Update the location $\vec{g}$ by (\ref{sp2}).
        		    \EndIf
                \Else
                    \State Update $\vec{D}$ by (\ref{bb1}) and $\vec{g}$ by (\ref{bb2}).
                \EndIf
        	\EndFor
        	\State Calculate the fitness of each search agent by (\ref{fitness}).
        	\State Update the location of the best search agent $\vec{g}^{*}(\hat{j})$.
        	\State Update $\hat{j} \leftarrow \hat{j}+1$.
        \Until{$\hat{j} > \hat{j}_{\max}$}
    \renewcommand{\algorithmicrequire}{\textbf{Output:}}
	\Require Optimal power control $\boldsymbol{{g}^*}$.
	\end{algorithmic}
\end{algorithm}
Conventionally, once the whales locate the prey, they approach it using either shrinking encircling method or spiral bubble-net method synchronously. To imitate this synchronous behaviour, we set the 50$\%$ probability to choose between these two methods to update the location of the whales for the optimization. Mathematically, it can be modeled as follows:
\begin{equation}
    \vec{g}(\hat{j}+1)= \begin{cases}\vec{g}^{*}(\hat{j})-\vec{A} \cdot \vec{D}, & \text { if } p<0.5, \\ \vec{D^{\prime}} \cdot e^{b j} \cdot \cos (2 \pi l)+\vec{g}^{*}(\hat{j}), & \text { if } p \geq 0.5.\end{cases}
\end{equation}
where $p = [0,1]$ is the random number to represent the probability to choose between two mechanisms. When $p<0.5$, WOA chooses the shrinking encircling mechanism, and if $p \geq 0.5$, WOA chooses the sprial movement mechanism. 

\subsubsection{Exploration State}: The exploration state of WOA includes the searching for prey method. This state is necessary to prevent the solution from being trapped at the local optimum, and failing to achieve the global optimum.

\textbf{Searching for Prey Method}. This approach is similar to encircling prey method in exploitation state, but instead of claiming the location of best search agent, and here, a random location is selected to update the locations of other search agents. It can be mathematically represented as follows:
\begin{equation}\label{sp1}
    \vec{D}=\left|\vec{C} \cdot \vec{g}_{\mathrm{rand}}(\hat{j})-\vec{g}(\hat{j})\right|,
\end{equation}
\begin{equation}\label{sp2}
    \vec{g}(\hat{j}+1)=\vec{g}_{\mathrm{rand}}(\hat{j})-\vec{A} \cdot \vec{D},
\end{equation}
where $\vec{g}_{\mathrm{rand}}(\hat{j})$ is the location of the search agent randomly selected from the search space.

Since WOA algorithm is designed only for unconstrained optimization, we apply the penalty method to our sub-problem $\textbf{P3}$ in order to deal with the minimum achievable date rate constraint (\ref{c31}) in the problem \cite{pham2020whale}. In our scenario, UEs are considered as a search agent, and the power control of the BS $\boldsymbol{g}$ represents the location of the search agents. At each iteration $\hat{j}$, the power control $\boldsymbol{g}$ can be updated by either the encircling prey method, spiral bubble-net attack method, or searching for prey method. The fitness function of our problem which chooses the optimal search agent can be expressed as follows:
\begin{equation}\label{fitness}
    \text{Fitness}(\boldsymbol{g}) = -\frac{R(\boldsymbol{g})}{P(\boldsymbol{g})} + \varpi \sum_{k=1}^{K} F_k(f_k(\boldsymbol{g}))f_k^2(\boldsymbol{g}),
\end{equation}
where $f_k(\boldsymbol{g})=r_k^{\min}[t]-r_k[t]$ is the inequality function, and $\varpi$ is the penalty factor coefficient. Since our sub-problem $\textbf{P3}$ is the maximization problem, we add the negative sign ahead of the objective function to convert into a minimization problem. The index function $F_k(f_k(\boldsymbol{g}))=1$ if $f_k(\boldsymbol{g})<0$, and $F_k(f_k(\boldsymbol{g}))=0$ if $f_k(\boldsymbol{g})\geq 0$. The pseudo-code of our WOA based power control can be described as in Algorithm \ref{algo3}.
\begin{algorithm}[t!] \caption{Proposed joint ARIS deployment, ARIS reflective elements on/off states, phase shift, power control optimization algorithm}\label{algo4}
    \small
	\begin{algorithmic}[1]
    	    \State {\textbf{Initialization:}} At $\tau = 0$, initialize the variables, $\boldsymbol{q}(0), \boldsymbol{\Delta}(0), \boldsymbol{\Theta}(0), \boldsymbol{g}(0)$;
    	    \Repeat
        	\State By applying $\textbf{Algorithm 1}$, solve problem $\textbf{P1}$ for given $\boldsymbol{\Delta}(\tau), \boldsymbol{\Theta}(\tau)$, $\boldsymbol{g}(\tau)$ to obtain $\boldsymbol{q}(\tau+1)$.
        	\State By applying $\textbf{Algorithm \ref{algo2}}$, solve problem $\textbf{P2}$ for given $\boldsymbol{q}(\tau+1), \boldsymbol{g}(\tau)$ to obtain $\boldsymbol{\Delta}(\tau+1), \boldsymbol{\Theta}(\tau+1)$.
            \State By applying $\textbf{Algorithm \ref{algo3}}$, solve problem $\textbf{P3}$ for given $\boldsymbol{q}(\tau+1), \boldsymbol{\Delta}(\tau+1), \boldsymbol{\Theta}(\tau+1)$ to obtain $\boldsymbol{g}(\tau+1)$.
        	\State Update $\tau \leftarrow \tau+1$.
            \Until{objective value (\ref{OF}) reaches convergence.}
	\end{algorithmic}
\end{algorithm}
\subsection{Overall Algorithm Complexity Analysis}
The overall iterative algorithm for solving our optimization problem (\ref{OF}) is described in Algorithm \ref{algo4} with the aforementioned proposed solutions to three sub-problems. According to the results in \cite{li2021robust, salman20223to} and \cite{pham2020whale}, the complexity of our solutions can be obtained by each algorithm for each sub-problem. For ARIS deployment sub-problem, the SCA is adopted as in Algorithm \ref{algo1}. Since there are $K$ users, the computational complexity of the SCA method is obtained as $\mathcal{O}_1\left(K^{3.5} \log(1/\varepsilon_1)\right)$ where $\varepsilon_1$ is the variable to control the accuracy of the SCA algorithm. For the AC-PPO algorithm for ARIS reflective elements on/off states and phase shift as in Algorithm \ref{algo2}, the computational complexity is $\mathcal{O}_2\left(a^2K\right)$, where $a \in \mathcal{A}$ is the total number of actions taken by the agent. With WOA for power control as in Algorithm \ref{algo3}, the computational complexity is $\mathcal{O}_3\left(\hat{J}U(m+K)\right)$, where $\hat{J}$ is the number of iterations for WOA, $U=30$ denotes the number of whale populations, and $m$ represents the number of inequality constraints in sub-problem $\textbf{P3}$. Henceforth, the overall computational complexity for solving (\ref{OF}) can be acquired as $\mathcal{O}\left(\hat{\tau}K^{3.5} \log(1/\varepsilon_1)+ \hat{\tau}a^2K + \hat{\tau}\hat{J}U(m+K) \right)$, where $\hat{\tau}$ denotes the number of iterations for Algorithm $\ref{algo4}$.

\section{Performance Evaluation}\label{evaluation}
In this section, we evaluate our proposed technique of energy-efficient multiple ARISs-assisted downlink communication system via numerical analysis. The network design comprises of 12 UEs uniformly distributed within 100 m $\times$ 100 m square region and the BS with 15 multiple antennas located at the center of the coverage area. There are 4 ARISs to support communication, and each ARIS is integrated with 10 reflective elements. The ARISs can hover at a maximum altitude of 100 m. The simulation parameters can be observed in Table \ref{simtable}. 
\begin{table}
	\centering
	\caption{Simulation parameters}
	\begin{tabular}{|l|c|}
		\hline
		\textbf{Parameter}                 & \textbf{Value}    \\
		\hline \hline
		Number of reflective elements on ARIS $n$, $I_n$  & 10 \\ \hline
		Bandwidth $W$		& 2 MHz \\ \hline
		Noise power $\sigma^2$		& -174 dBm \\ \hline
		Path loss exponent $\alpha$                & 4 \\ \hline
		Channel gain at reference distance $\kappa$ & -40 dBm \\ \hline
		Rician factor $\hat{R}$                     & 10 \\ \hline
		Circuit power of each RIS element $P_{\mathrm{RIS}}$ & 10 dBm \cite{yang2021energy} \\ \hline
		BS power amplifier efficiency $\mu$ & 0.8 \cite{yang2021energy} \\ \hline
		Circuit power of each user $P_{k}^{\mathrm{cir}}$        & 10 dBm \cite{yang2021energy} \\ \hline
		Clipping parameter $\epsilon$ & 0.2 \\ \hline
		Learning rate & 0.0002 \\ \hline
		Discount factor $\xi$ & 0.9 \\ \hline
		Mini batch size & 64 \\ \hline
		Number of episodes & 1,000 \\ \hline
		Number of time steps     & 300,000       \\ \hline
	\end{tabular}
	\label{simtable}
\end{table}
To evaluate our proposed algorithm, we compare our method with four benchmark schemes, which are explained as follows:
\begin{itemize}
    \item \emph{Single-ARIS:} In this scheme, we implement a single ARIS instead of using multiple ARISs to support the downlink communications from BS. The optimization problem is then solved by using our proposed algorithms.
    \item \emph{ARIS (NPS):} In this approach, we deploy 4 ARISs with the fixed phase shifts. The ARIS deployment problem is solved by SCA, the ARIS reflective elements on/off states problem is solved by AC-PPO, and transmit power allocation problem is solved by WOA alternatively.
    \item \emph{Random:} In this design, we randomly deploy the ARISs, fixed the reflective elements ON/OFF states, and fixed the transmit power of the BS.
    \item \emph{UAV-Relay:} In this method, ARIS is not used. Instead, 4 UAVs are deployed as relays and the incident signal is linearly processed, and re-transmit them toward the required destination. The optimization problem is then solved by using our proposed algorithms.
\end{itemize}

\begin{figure}
\begin{minipage}[t]{0.49\linewidth}
    \centering
	\includegraphics[width=1\textwidth]{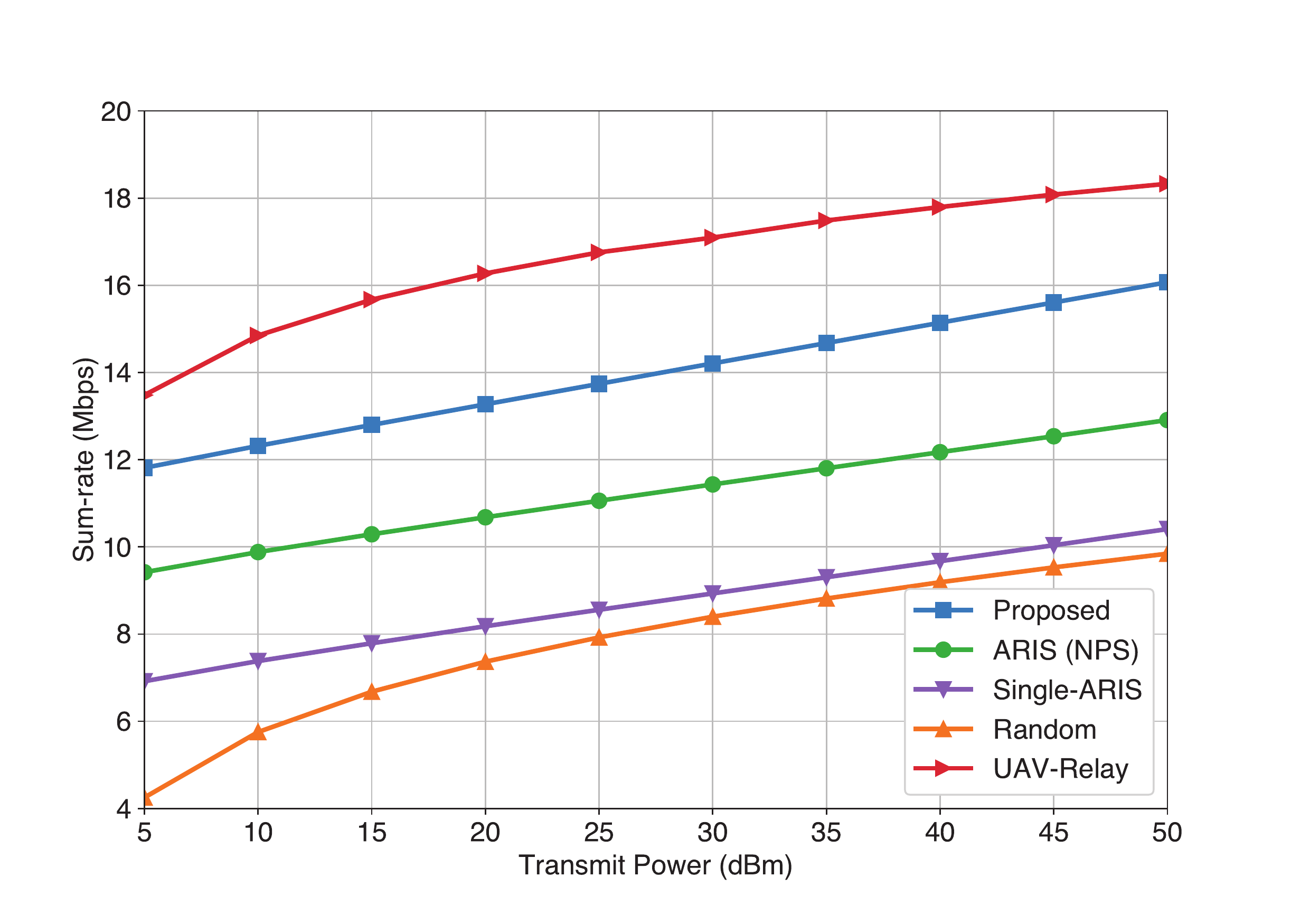}
	\caption{Performance comparison of sum-rate for different transmit power.}
	\label{srvstp}
\end{minipage}
\hspace{0.1cm}
\begin{minipage}[t]{0.49\linewidth}
    \centering
	\includegraphics[width=1\textwidth]{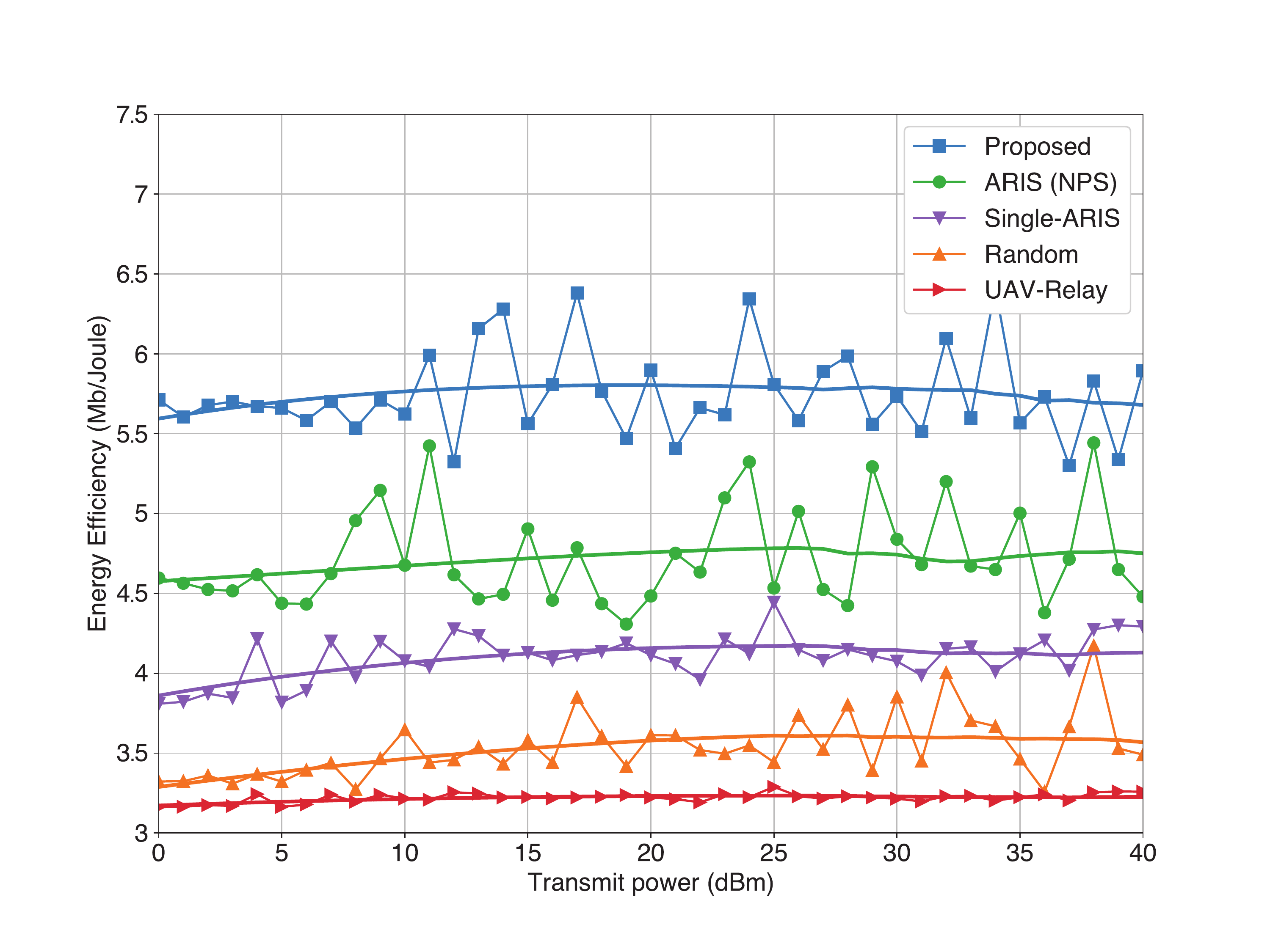}
	\caption{Performance comparison of energy efficiency for different transmit power.}
	\label{eevstp}
\end{minipage}
\end{figure}
	
\begin{figure}
\begin{minipage}[t]{0.49\linewidth}
    \centering
	\includegraphics[width=1\textwidth]{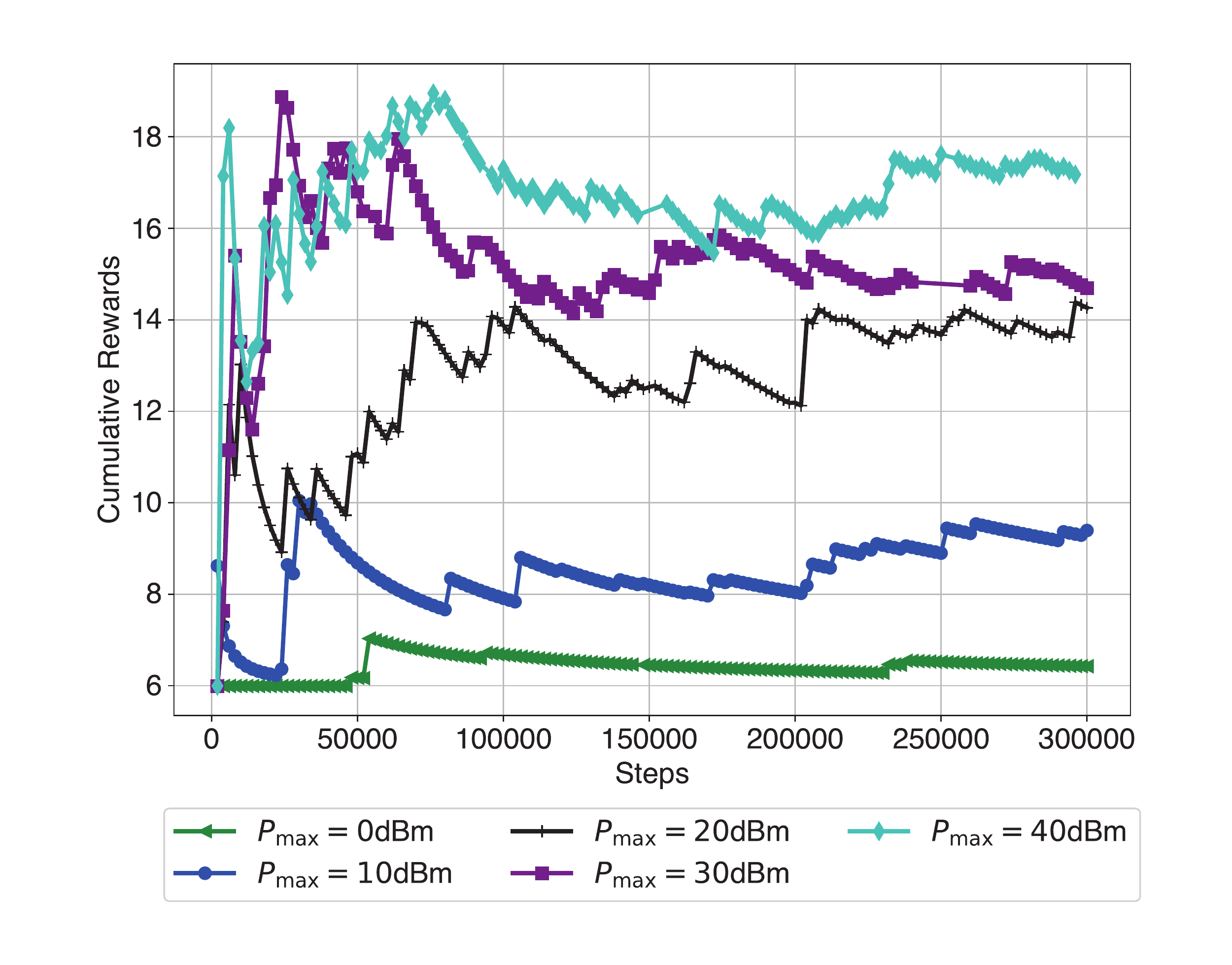}
	\caption{Performance comparison of cumulative rewards for different $P_{\max}$.}
	\label{reward_pmax}
\end{minipage}
\hspace{0.1cm}
\begin{minipage}[t]{0.49\linewidth}
    \centering
	\includegraphics[width=1\textwidth]{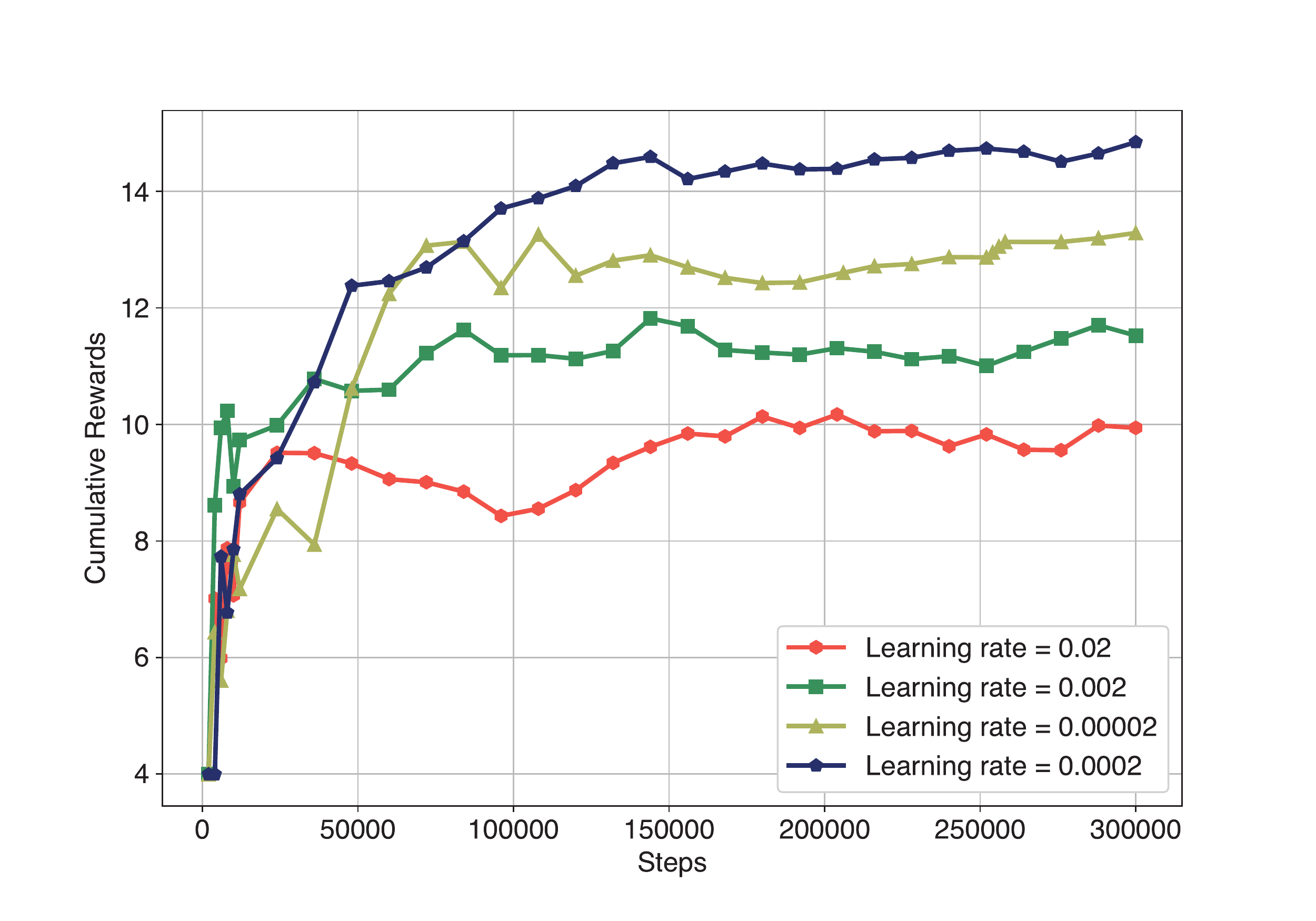}
	\caption{Performance comparison of cumulative rewards for different learning rate.}
	\label{lr}
\end{minipage}
\end{figure}

\begin{figure}
\begin{minipage}[t]{0.5\linewidth}
    \centering
	\includegraphics[width=1\textwidth]{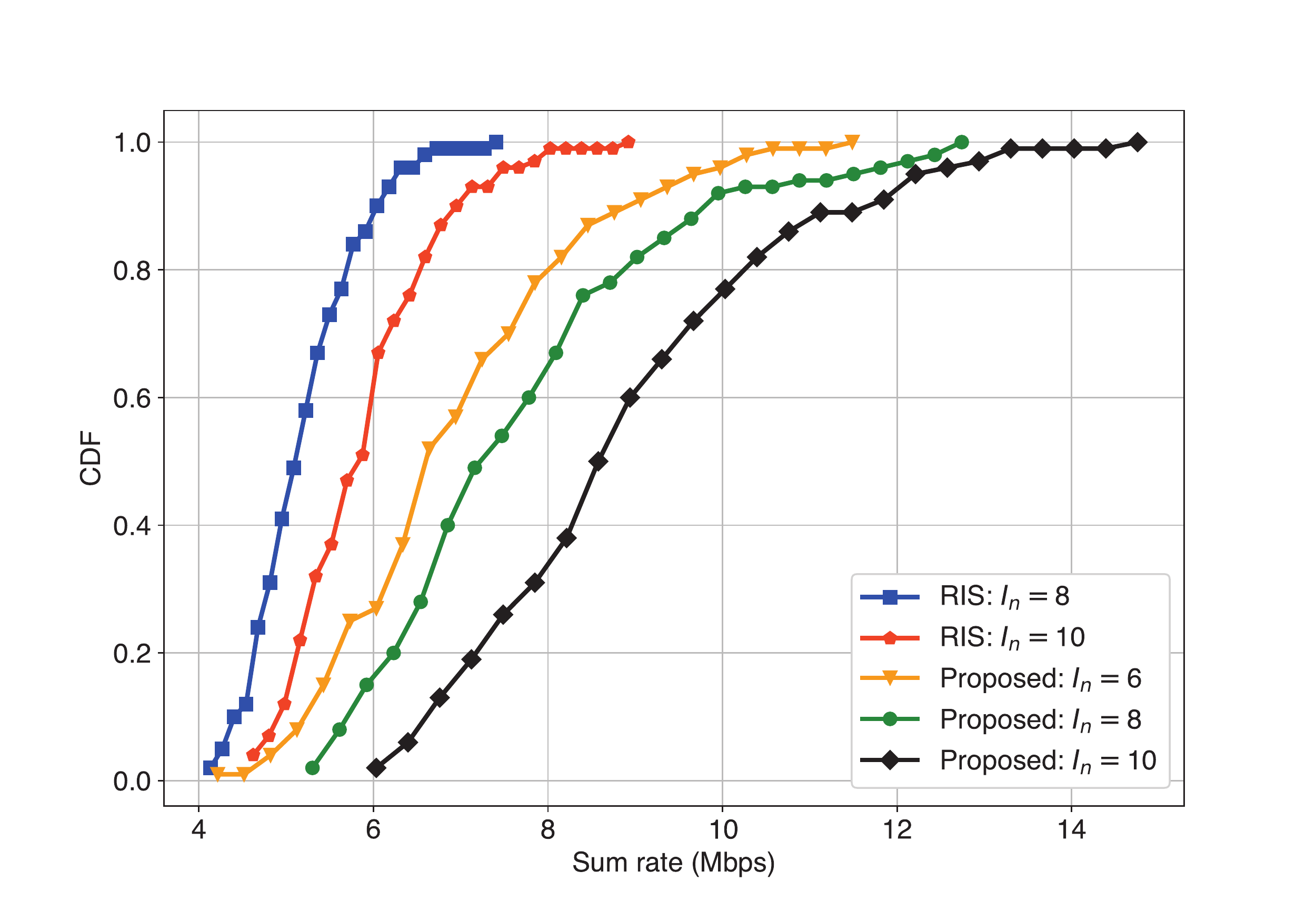}
	\caption{CDF of sum-rate with different number of ARIS reflective elements for RIS and proposed system.}
	\label{srcdf}
\end{minipage}
\hspace{0.1cm}
\begin{minipage}[t]{0.5\linewidth}
    \centering
	\includegraphics[width=1\textwidth]{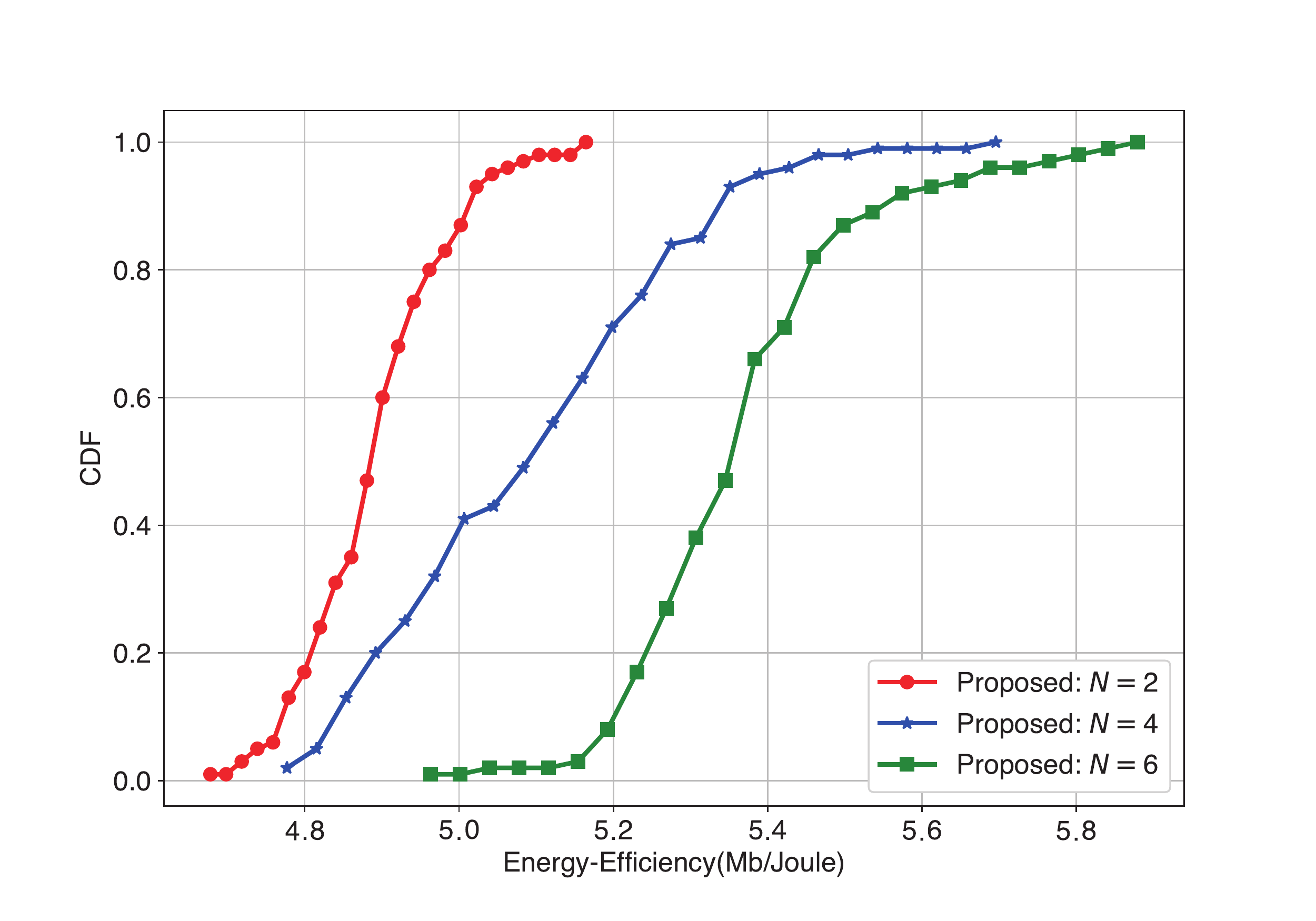}
	\caption{CDF of energy-efficiency with different number of ARIS.}
	\label{eecdf}
\end{minipage}
\end{figure}

\begin{figure}
    \centering
	\includegraphics[width=0.5\linewidth]{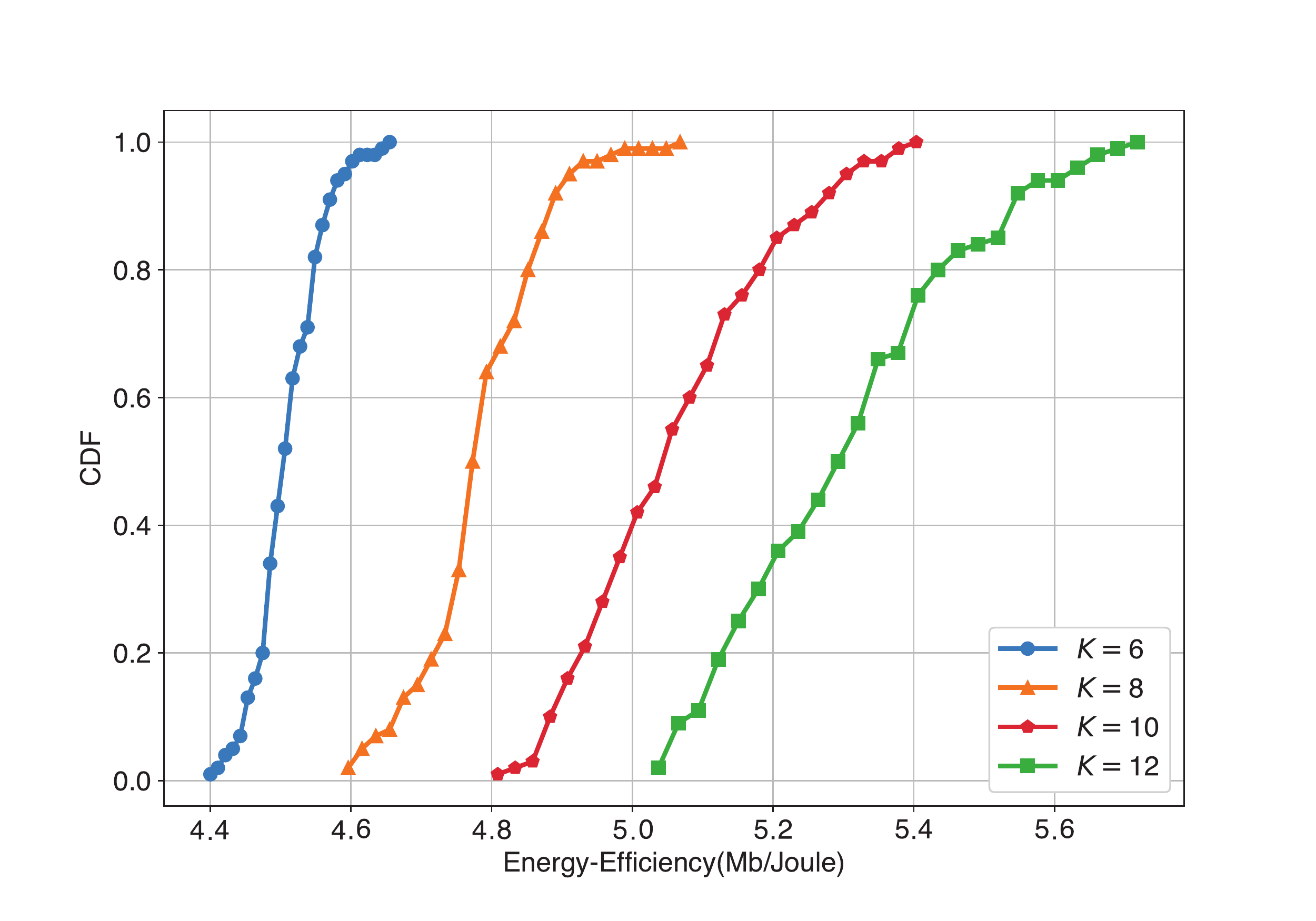}
	\caption{CDF of energy-efficiency with different number of UEs.}
	\label{eecdfwuser}
\end{figure}

Fig. \ref{srvstp} compares the average sum-rate of users with our proposed DRL-based algorithm towards the above-mentioned benchmark schemes. In all circumstances, the average sum-rate increases as the maximum transmit power rises. Our proposed algorithm outperforms ARIS (NPS) by $24\%$ and single-ARIS by $58\%$, respectively. This demonstrates how the multiple ARISs can achieve better outcomes than a single ARIS since it can provide several paths between the BS and UEs. Our algorithm outperforms most benchmark schemes in average sum-rate except for the UAV-relay scenario. UAV-relay provides the highest performance since it processes and re-transmits the incident signal using a dedicated power source. As a consequence, it consumes more energy which can be observed in Fig. \ref{eevstp}.

Fig. \ref{eevstp} depicts the comparison of the energy efficiency under different algorithms. The smooth data is demonstrated by the solid curved line, which represents the Savitzky-Golay filter. In all scenarios, energy efficiency increases faster until the transmit power of the BS reaches to $10$ dBm. Since then, the energy efficiency hasn't improved much as the function does not increase monotonically with respect to to transmit power. Our proposed algorithm achieves $72\%$ increase compared to the UAV-relay scenario and $43\%$ increase compared to the single-ARIS scenario.

Next, we evaluate the convergence of our proposed AC-PPO algorithm with different values of $P_{\max}$ ranging from 0 dBm to 40 dBm. As shown in Fig. \ref{reward_pmax}, it can be observed that in all cases, the convergence of our cumulative rewards increases with respect to increase in transmit power. We can see a significant difference in the performance when $P_{\max}$ is low, and the performance difference becomes lesser as $P_{\max}$ becomes higher. This suggests that SINR has significant impact on the overall performance of the cumulative rewards.

Following that, we examine how various learning rates affect on our cumulative rewards, ranging from the set of $\{0.02, 0.002, 0.0002, 0.00002\}$. As seen in Fig. \ref{lr}, a higher learning rate does not enable our cumulative reward to converge faster but provides less performance. Although it takes longer to converge, the learning rate of 0.0002 delivers better performance than 0.002 and 0.00002. In this case, we chose a learning rate of 0.0002 since it produces the highest cumulative rewards for our proposed method.

Furthermore, we compare the spectral efficiency of our proposed multiple ARISs-assisted system to that of multiple RISs-assisted systems. In this approach, we employ 4 RISs on the ground level rather than mounted on the UAVs. Fig. \ref{srcdf} demonstrates the cumulative distribution function (CDF) values of average sum-rates. As seen in Fig. \ref{srcdf}, our proposed system achieved $69\%$ performance increase compared to RISs-assisted system. This is because our proposed system takes into account the deployment of UAVs, which provides improved LOS communications between the BS and UEs. Concurrently, we experiment the performance of spectral efficiency with different numbers of reflective elements. In all scenarios, the results show that as the number of reflective elements increases, so do the UEs' average sum-rates. This indicates that the spectral efficiency will be improved by increasing the number of reflective parts.

Finally, we examine the energy efficiency with various ARIS numbers and different numbers of UEs, respectively. As shown in Fig. \ref{eecdf}, we can observe that the energy efficiency improves as the number of ARIS components increases. Moreover, when the number is low, the energy efficiency improvement is significantly more compared to a larger number of ARIS. This indicates that for the small cell network with 12 UEs, we do not need to install a large amount of ARISs. Next, as shown in Fig. \ref{eecdfwuser}, we can observe that the energy efficiency almost linearly increases with increasing number of UEs between 6 to 12. We can perceive that more ARISs and more UEs help improve the energy efficiency of the multiple ARISs-assisted downlink communication system.

\section{Conclusion}\label{conclusion}
In this paper, we have studied an energy-efficient multiple ARISs-assisted downlink communication system. To maximize energy efficiency, we formulated a joint ARIS deployment, ARIS reflective elements on/off states, phase shift, and power control problem. As the formulated problem is MINLP and NP hard, we decompose our problem into three sub-problems: ARIS deployment problem, joint reflective elements ON/OFF states and phase shift problem, and power control problem. We then proposed SCA approach, AC-PPO method and WOA to solve our sub-problems, alternatively. Through extensive numerical analysis, we have proved that by integrating multiple ARISs in the downlink communication system, it can significantly outperform several benchmark schemes; especially in spectral efficiency compared to a multiple RISs-assisted scenario and energy efficiency compared to a single ARIS-assisted scenario.

\numberwithin{equation}{section}
\appendices
\ifCLASSOPTIONcaptionsoff
  \newpage
\fi



%


\bibliographystyle{IEEEtran}
\bibliography{mybib}

%








\end{document}